\documentclass[12pt]{amsart}

\usepackage{amsmath, amsthm, amssymb}
\usepackage[colorlinks=true]{hyperref}

\theoremstyle{plain} 
\newtheorem{thm}[equation] {Theorem}
\newtheorem{conj}[equation]{Conjecture}
\newtheorem{cor}[equation]{Corollary}
\newtheorem{lem}[equation]{Lemma}
\newtheorem{prop}[equation]{Proposition}

\theoremstyle{definition}
\newtheorem{defn}[equation]{Definition}
\newtheorem{ex}[equation]{Example}

\theoremstyle{remark}
\newtheorem{rem}[equation]{Remark}

\numberwithin{equation}{section}
\numberwithin{figure}{section}

\DeclareMathOperator{\Hom}{Hom}
\DeclareMathOperator{\sym}{Sym}
\newcommand{\ra}{\rightarrow}
\newcommand{\field}[1]{\mathbb{#1}}
\newcommand{\F}{\ensuremath{\field{F}}}
\newcommand{\C}{\ensuremath{\field{C}}}
\newcommand{\N}{\ensuremath{\field{N}}}
\newcommand{\Q}{\ensuremath{\field{Q}}}
\newcommand{\Z}{\ensuremath{\field{Z}}}
\newcommand{\size}[1]{\lvert #1 \rvert}
\newcommand{\abs}[1]{\lvert #1 \rvert}
\newcommand{\al}{\alpha}
\newcommand{\be}{\beta}
\newcommand{\la}{\lambda}
\newcommand{\OO}{\ensuremath{\mathcal O}}
\newcommand{\U}{\ensuremath{\mathcal{U}}}
\DeclareMathOperator{\tr}{tr}
\DeclareMathOperator{\hwe}{hwe}
\DeclareMathOperator{\ewe}{ewe}
\newcommand{\wh}{{\textsc{h}}}
\newcommand{\wlee}{{\textsc{l}}}
\newcommand{\we}{{\textsc{e}}}
\DeclareMathOperator{\wwe}{wwe}
\newcommand{\asing}{A^{\rm sing}}
\newcommand{\adble}{A^{\rm dble}}
\DeclareMathOperator{\se}{se}
\newcommand{\wl}{w^{(\ell)}}
\DeclareMathOperator{\efflength}{efflen}

\title[Weights on finite fields
]{Weights on finite fields and \\ failures of the MacWilliams identities}

\author[J. A. Wood]{Jay A. Wood}
\address{
Western Michigan University
}
\email{jay.wood@wmich.edu}

\dedicatory{In memoriam: \\
William Browder, \oldstylenums{1934}--\oldstylenums{2025} \\
Alan T. Huckleberry, \oldstylenums{1941}--\oldstylenums{2025} \\
Catherine M. Murphy, \oldstylenums{1940}--\oldstylenums{2025}  }



\begin{document}

\begin{abstract}
In the 1960s, MacWilliams proved that the Hamming weight enumerator of a linear code over a finite field completely determines, and is determined by, the Hamming weight enumerator of its dual code.  In particular, if two linear codes have the same Hamming weight enumerator, then their dual codes have the same Hamming weight enumerator.

In contrast, there is a wide class of weights on finite fields whose weight enumerators have the opposite behavior: there exist two linear codes having the same weight enumerator, but their dual codes have different weight enumerators.
\end{abstract}

\maketitle

\section{Introduction}

In the early 1960s, Florence Jessie MacWilliams (1917--1990) proved that the Hamming weight enumerator of a linear code $C \subseteq \F_q^n$ determines the Hamming weight enumerator of the dual code $C^\perp$, and vice versa.  The relationship between the weight enumerators is encapsulated in the \emph{MacWilliams identities} \cite{MR2939359, MR0149978}:
\begin{equation}   \label{eq:MWforFq}
\hwe_{C^\perp}(X,Y) = \frac{1}{\size{C}} \hwe_C(X+ (q-1)Y, X-Y) .
\end{equation}
An immediate consequence of the MacWilliams identities is:  If two linear codes $C, D \subseteq \F_q^n$ satisfy $\hwe_C = \hwe_D$, then $\hwe_{C^\perp} = \hwe_{D^\perp}$.  We say that the Hamming weight \emph{respects duality}.

The MacWilliams identities have been generalized in several ways: by generalizing the alphabets for the codes and by generalizing the enumerators.  For the Hamming weight enumerator, the MacWilliams identities hold for additive codes over a finite abelian group \cite{MR384310} and for linear codes over a finite Frobenius ring \cite{wood:duality}.

The Hamming weight enumerator can be interpreted in two ways.  One way is to \emph{count} the number of nonzero entries in a codeword; the other is to \emph{add} the values of the weight of the entries in a codeword.  These interpretations lead to two generalizations of the Hamming weight enumerator: by using more general partitions or by using more general weights.

Let $A$ be a finite abelian group, and suppose $\mathcal{P}$ is a partition of $A$.  For any additive code $C \subseteq A^n$, one can define a partition enumerator that encodes the numbers of entries in a codeword that belong to the various blocks of the partition.  The Hamming weight enumerator is the partition enumerator for the partition of $A$ into $\{0\}$ and $A-\{0\}$, while the complete enumerator uses the complete partition $A = \cup_{a \in A} \{a\}$ into singletons.  Certain partitions, called reflexive partitions, give rise to well-behaved MacWilliams identities.  For this, and more, see \cite{MR3336966}.

The other generalization, using more general weights, is the topic of this paper.  Suppose $A$ is a finite abelian group equipped with an integer-valued weight $w$.  We assume $w \colon  A \ra \Z$, with $w(0)=0$, and $w(a) > 0$ for $a \neq 0$.  Denote the maximum value of $w$ by $w_{\max}$.  Extend $w$ additively to $A^n$, so that $w(a_1, a_2, \ldots, a_n) = \sum_{i=1}^n w(a_i)$.  For any additive code $C \subseteq A^n$, define its \emph{$w$-weight enumerator} by 
\[  \wwe_C(X,Y) = \sum_{c \in C} X^{n w_{\max} - w(c)} Y^{w(c)} . \]
It is natural to ask: \emph{Do MacWilliams identities hold for the $w$-weight enumerator?}

The MacWilliams identities hold for the Hamming weight, as described above.  They also hold for the homogeneous weight over $\F_q$ (because the homogeneous weight is a multiple of the Hamming weight in that setting), for the Lee weight over $\Z/4\Z$ \cite{Hammons-etal}, and for the homogeneous weight over $M_{2 \times 2}(\F_2)$ \cite{MR4963840}.  The MacWilliams identities fail to hold, by virtue of the weight failing to respect duality: for the Rosenbloom-Tsfasman weight on matrices \cite{MR1900585}; for the Lee weight over $\Z/m\Z$, $m \geq 5$ \cite{MR4119402}; for the homogeneous weight over $\Z/m\Z$, $m \geq 6$, $m$ not prime \cite{revuma.2807}; and for many weights having maximal symmetry over finite chain rings or over matrix rings $M_{k \times k}(\F_q)$ \cite{MR4963840}.  The latter results suggest that it might be a rare event for a weight to respect duality.

The present paper reinforces the rarity of a weight respecting duality.  The main result, Theorem~\ref{thm:MainV2}, can be interpreted as saying that, for weights on finite fields subject to certain restrictions, a generic weight does not respect duality.

Here is an outline of the argument.  Consider an integer-valued weight $w$ on $\F_q$ with the restrictions that $w(a) = w(-a)$ for $a \in \F_q$ and that $w$ is not a multiple of the Hamming weight.  Assume also that $w$ is nondegenerate (defined precisely in Definition~\ref{def:nondegenerate}).  The goal is to prove that $w$ does not respect duality by constructing linear codes $C, D \subseteq \F_q^n$, for some $n$, such that $\wwe_C = \wwe_D$, yet $\wwe_{C^\perp} \neq \wwe_{D^\perp}$.

We consider linear codes of dimension $2$ over $\F_q$, and we consider, without loss of generality, generator matrices of size $2 \times n$ whose columns have a specific normal form  \eqref{eqn:OOReps} that reflects the symmetry of the weight $w$.  To present a linear code $C$, it is then enough to specify a multiplicity function $\eta_C$ that keeps track of how many times a specific column appears in the generator matrix for $C$.  An ordered listing $\omega_C$ of the weights of the codewords in $C$ is then a linear function of the multiplicity function: $\omega_C = W \eta_C$, where the matrix $W$ depends only on the weight $w$.

Two linear codes $C, D \subseteq \F_q^n$ will satisfy $\wwe_C = \wwe_D$ if and only if their ordered lists of weights of codewords, $\omega_C$ and $\omega_D$, are permutations of each other.  At this point, we work backwards---start with ordered lists $\omega_C$ and $\omega_D$ that are permutations of each other, and invert the linear equations $\omega = W \eta$ to find $\eta_C$ and $\eta_D$.  The matrix $W$ has a lot of structure, which allows for careful analysis.  By keeping the listings $\omega_C$ and $\omega_D$ very simple, with just two different weights, it is possible to prove that the resulting multiplicity functions, $\eta_C$ and $\eta_D$, have nonnegative entries, and so can be used to define linear codes.  It is also possible to analyze the numbers of dual codewords of small weight and to show that they are different, at least generically.

\subsubsection*{Acknowledgments}
The year 2025 marked the passing of three mathematicians who influenced me greatly.  Professor Huckleberry taught complex analysis when I was an undergraduate, and his courses encouraged me to study complex differential geometry in graduate school.  Professor Browder suggested I study $\Z/4\Z$-valued binary quadratic forms, which led me into the beautiful $\Z_4$-world in coding theory.  Cathy Murphy was my colleague and department head for ten years.  Her generous and thoughtful mentorship opened up many opportunities for me.

\section{Preliminaries}  \label{sec:Prelims}

In this section we review some definitions from algebraic coding theory and establish notation that will be used throughout the paper.  Because most of the theorems we will prove are for linear codes defined over a finite field $\F_q$ or over an integer residue ring $\Z/m\Z$, we will confine our definitions to linear codes defined over a finite commutative ring with $1$.  For additional details or the modifications needed for noncommutative rings, see \cite{wood:turkey}.

Let $R$ be a finite commutative ring with $1$.  Denote by $\U(R)$ the group of units, i.e., invertible elements, of $R$.  The  group of units of a finite field $\F_q$ equals the multiplicative group $\F_q^\times$ of nonzero elements; $\F_q^\times$ is known to be a cyclic group, generated by some element $\al$, called a \emph{primitive element} of $\F_q$. 
Define the \emph{dot product} on $R^n$ by
\[  x \cdot y = \sum_{i=1}^n x_i y_i \in R, \]
for $x = (x_1, x_2, \ldots, x_n), y=(y_1, y_2, \ldots, y_n) \in R^n$.  

A \emph{linear code} of length $n$ over $R$ is an $R$-submodule $C \subseteq R^n$.  Given a linear code $C \subseteq R^n$, define its \emph{dual code} $C^\perp \subseteq R^n$ by
\[  C^\perp = \{ y \in R^n \colon  x \cdot y = 0 \text{ for all $x \in C$}\} . \]
In general, for a linear code $C \subseteq R^n$, we have $C \subseteq (C^\perp)^\perp$.  When $R$ is a finite field $\F_q$ and $C \subseteq \F_q^n$ is a linear code, then linear algebra tells us that $\dim_{\F_q} C^\perp = n - \dim_{\F_q} C$ and $C = (C^\perp)^\perp$.  These properties generalize to finite Frobenius rings \cite{wood:turkey}: $\size{C} \cdot \size{C^\perp} = \size{R^n}$, and $C = (C^\perp)^\perp$.  Finite fields are Frobenius.

A \emph{weight} on $R$ is a complex-valued function $w \colon  R \ra \C$ with $w(0)=0$.  In this paper, weights will be assumed to have integer values, with $w(0)=0$ and $w(r) >0$ for $r \neq 0$.  Denote the maximum value of $w$ by $w_{\max}$ and the smallest positive value of $w$ by $\mathring{w}$. A weight $w$ determines a \emph{symmetry group} $\sym(w)$ in $\U(R)$:
\begin{equation}  \label{eqn:DefnSymGroup}
\sym(w) = \{ u \in \U(R) \colon  w(ur) = w(r) \text{ for all $r \in R$}\} .
\end{equation}

The symmetry group $\sym(w)$ acts on $R$ by ring multiplication, and $\sym(w)$ acts on $R^\sharp = \Hom_R(R,R)$, the group of $R$-module homomorphisms $\la \colon R \ra R$, by multiplication on the values of $\la \in R^\sharp$.  Elements $\la \in R^\sharp$ are called \emph{functionals}; their inputs are written on the left.  Denote the $\sym(w)$-orbits of $x \in R$ and $\la \in R^\sharp$ by $[x]$ and $[\la]$, respectively.  Define a matrix $\mathcal{W}$ with rows indexed by the nonzero $\sym(w)$-orbits of $R$ and columns indexed by the nonzero $\sym(w)$-orbits of $R^\sharp$:
\begin{equation}  \label{eq:R-W}
\mathcal{W}_{[x], [\la]} = w(x \la) .
\end{equation}
These values are well-defined, given the definition of $\sym(w)$.

\begin{defn}  \label{def:nondegenerate}
A weight $w$ on $R$ is \emph{nondegenerate} if the matrix $\mathcal{W}$ of \eqref{eq:R-W} is invertible over $\Q$.
\end{defn}

\begin{ex}
The \emph{Hamming weight} is defined over any finite ring $R$:
\[  \wh(r) = \begin{cases}
0, & r = 0, \\
1, & r \neq 0.
\end{cases} \]
The symmetry group of the Hamming weight is $\sym(\wh) = \U(R)$.  The Hamming weight is nondegenerate if and only if $R$ is Frobenius: \cite[Theorems~6.3, 6.4]{wood:duality} and \cite[Theorem~7.2]{wood:turkey}.
\end{ex}

\begin{ex}  \label{ex:PowerWeights}
For the integer residue rings $R = \Z/m\Z$, represent residue classes uniquely with integers $r$ satisfying $-m/2 < r \leq m/2$.  Define the Lee weight $\wlee$ and Euclidean weight $\we$ on $\Z/m\Z$ by
\[  \wlee(r) = \abs{r} , \quad \we(r) = \abs{r}^2 , \]
where $\abs{r}$ is the ordinary (archimedean) absolute value on $\Z$.  The Hamming, Lee, and Euclidean weights are cases $\ell=0, 1, 2$ of the \emph{power weights} $\wl$ defined for integers $\ell \geq 0$ by
\[  \wl(r) = \abs{r}^\ell , \quad r \in \Z/m\Z . \]
Except for the Hamming weight, all the power weights $\wl$, $\ell = 1, 2, \ldots$, have $\sym(\wl) = \{ \pm 1\} \subseteq \U(\Z/m\Z)$.  The Lee and Euclidean weights are known to be nondegenerate \cite{MR3947347}.
\end{ex}

While not studied in this paper, the homogeneous weight is another important weight defined on any finite ring with $1$;  it is nondegenerate if and only if the ring is Frobenius \cite{greferath-schmidt:combinatorics}.

A weight $w$ on $R$ extends additively to a weight on $R^n$ by
\[  w(r_1, r_2, \ldots, r_n) = \sum_{i=1}^n w(r_i) . \]
There are other ways to define weights on $R^n$, but they are not the subject of this paper.

The \emph{$w$-weight enumerator} of a linear code $C \subseteq R^n$ is a homogeneous polynomial of degree $n w_{\max}$ in $\C[X,Y]$ defined by
\[  \wwe_C(X,Y) = \sum_{c \in C} X^{n w_{\max} - w(c)} Y^{w(c)} . \]
To save space in displays, we will often set $X=1$ and $Y=y$.  This loses direct information about the length $n$ of the code.

For the Hamming weight on a finite field $\F_q$, MacWilliams \cite{MR2939359, MR0149978} proved that the Hamming weight enumerator of a linear code $C \subseteq \F_q^n$ is related to the Hamming weight enumerator of its dual; see \eqref{eq:MWforFq}.
This result was later generalized to the Hamming weight on finite Frobenius rings \cite{wood:duality}; such rings include the integer residue rings $\Z/m\Z$.

\begin{defn}
A weight $w$ on $R$ \emph{respects duality} when it has the following property: for linear codes $C, D \subseteq R^n$, if $\wwe_C = \wwe_D$, then $\wwe_{C^\perp} = \wwe_{D^\perp}$.
\end{defn}
The MacWilliams identities \eqref{eq:MWforFq} imply that the Hamming weight on $\F_q$ respects duality.

\begin{rem}
Another way to view respecting duality uses the weight distribution of a code.  For a code $C \subseteq R^n$ and an integer $j$, define 
\[  A_j(C) = \size{\{ c \in C \colon  w(c)=j \} } . \]
Knowledge of the $A_j$ determines $\wwe_C$:
\[  \wwe_C(X,Y) = \sum_j A_j(C) X^{n w_{\max}-j} Y^j . \]
A weight $w$ respects duality if, for linear codes $C, D$ of the same length, $A_j(C) = A_j(D)$ for all $j$ implies that $A_j(C^\perp) = A_j(D^\perp)$ for all $j$.
\end{rem}

The main result of this paper is that most weights on $\F_q$ do not respect duality.  
In order to state the result precisely, we will need a few more definitions.

Let $w$ be a weight on a finite field $\F_q$.  Assume that $w$ has integer values, with $w(0)=0$ and $w(r) >0$ for nonzero $r \in \F_q$.  Write $\mathring{w} = \min\{ w(r)  \colon  r \in \F_q, r \neq 0\}$.  Assume that the symmetry group $\sym(w)$ satisfies $-1 \in \sym(w)$ and $\sym(w) \neq \F_q^\times$.  (The last condition means that $w$ is not a multiple of the Hamming weight.)  Fix a primitive element $\al$ of $\F_q$; i.e., $\al$ generates the cyclic group $\F_q^\times$.  Let $t$ be the smallest positive integer such that $\al^t$ generates $\sym(w)$.  Then $t \size{\sym(w)} = \size{\F_q^\times} = q-1$.  Because $\sym(w) \neq \F_q^\times$, we have $t \geq 2$.
As $\al^t \in \sym(w)$, we see that $w(\al^{t+j}) = w(\al^t \al^j) = w(\al^j)$ for all $j$. 

Define $\breve{w}$ and $c_m(w,w)$, for $m=0, 1, \ldots, t-1$, by 
\[  \breve{w} = \sum_{j=0}^{t-1} w(\al^j) , \quad
c_m(w,w) = \sum_{j=0}^{t-1} w(\al^j) w(\al^{m+j}) .  \]

Here is the precise statement of the main result.
\begin{thm}  \label{thm:MainV2}
Let $w$ be an integer-valued weight on $\F_q$, with $w(0)=0$ and $w(r)>0$ for $r \neq 0$.  Assume $w$ satisfies:
\begin{itemize}
\item  $w$ is nondegenerate,
\item $-1 \in \sym(w)$,
\item $\sym(w) \neq \F_q^\times$,
\item  the minimum value $\mathring{w}$ is achieved on exactly one $\sym(w)$-orbit,
\item  at least two of the following values are different:
\[ qc_1(w,w), qc_2(w,w), \ldots, qc_{\lfloor t/2 \rfloor}(w,w), ((q-1)/t) \breve{w}^2 . \]
\end{itemize}
Then, for some $n$, there exist two linear codes $C,D \subseteq \F_q^n$, with $\wwe_C = \wwe_D$, but $\wwe_{C^\perp} \neq \wwe_{D^\perp}$.  In particular, $w$ does not respect duality.
\end{thm}

\section{Presenting linear codes via multiplicity functions}  \label{sec:MultFcns}

We continue to assume that $R$ is a finite commutative ring with $1$ and that $w$ is a weight on $R$, with $w(0)=0$ and $w(r)>0$ for all nonzero $r \in R$.  Extend $w$ additively to $R^n$, so that $w(r_1, r_2, \ldots, r_n) = w(r_1) + w(r_2) + \cdots + w(r_n)$.  The symmetry group $\sym(w)$ of $w$ was defined in \eqref{eqn:DefnSymGroup}.

An $R$-linear code $C \subseteq R^n$ can be presented as the row space of a $k \times n$ matrix $G$, called a \emph{generator matrix}.  Every codeword in $C$ then has the form $c = v G$, where $v \in R^k$.
Writing $v=(v_1, v_2, \ldots, v_k) \in R^k$ and $G = (g_{i,j})$, with $i=1, 2, \ldots, k$ and $j = 1, 2, \ldots, n$, then, for $c = vG$,
\begin{equation}  \label{eqn:wOfc}
w(c) = w(vG) = \sum_{j=1}^n w(\sum_{i=1}^k v_i g_{i,j}) .
\end{equation}
This formula for $w(c)$ remains the same if
\begin{itemize}
\item the columns of $G$ are permuted, or
\item for any $j=1, 2, \ldots, n$, the $j$th column of $G$ is multiplied by $h_j \in \sym(w)$.
\end{itemize}
Thus, there is no loss of generality to restrict the columns of $G$ to come from a set of representatives of the orbits of $\sym(w)$ acting on the space of all columns.  Furthermore, the values of $w(c)$ depend only the multiplicities of how many times a particular representative appears as a column of $G$.  Indeed, suppose the set of representatives is $\OO = \{\la_1, \la_2, \ldots, \la_N \}$, and that $\la_\iota$ appears $\eta_\iota$ times as a column of $G$.  Then \eqref{eqn:wOfc} becomes
\begin{equation}  \label{eqn:W-times-eta}
w(c) = \sum_{\iota=1}^N w(v \la_\iota) \eta_\iota .
\end{equation}

The symmetry group appears a second way, this time acting on the inputs $c = vG$.  If $u \in \sym(w)$, then $w(uvG) = w(vG)$.  As a consequence, the weight distribution of a linear code $c$ is completely determined by the values of $w$ on representatives of the $\sym(w)$-orbits $[v]$ in $R^k$, together with the sizes of the $\sym(w)$-orbits.  Given a representative $v$, its orbit contributes $\size{[v]}$ to $A_{w(vG)}(C)$.

We now work out the details for the situation where $R = \F_q$ is a finite field and $w$ is any integer-valued weight on $\F_q$, subject to a few restrictions.  As above, assume $w(0)=0$ and $w(r)>0$ for all nonzero $r \in \F_q$.  Also assume $-1 \in \sym(w)$ and $\sym(w) \neq \F_q^\times$, so that $w$ is not equal to a constant multiple of the Hamming weight.

As discussed in Section~\ref{sec:Prelims}, let $\al$ be a primitive element of $\F_q$, and let $t$ be the smallest positive integer such that $\al^t$ generates $\sym(w)$.  Then $t \size{\sym(w)} = \size{\F_q^\times} = q-1$.  As $\sym(w) \neq \F_q^\times$, it follows that $t \geq 2$.

We will construct linear codes over $\F_q$ having dimension $2$, so a generator matrix will have size $2 \times n$.  In particular, the columns of a generator matrix will belong to $\F_q^2$.  The symmetry group $\sym(w)$ acts on $\F_q^2$ by scalar multiplication.  Note that the $\sym(w)$-orbit of a nonzero vector $v \in \F_q^2$ will be contained in the $1$-dimensional vector subspace spanned by $v$.  To write down the set $\OO$ of representatives of the nonzero $\sym(w)$-orbits, we start by writing down certain elements that generate the $1$-dimensional subspaces of $\F_q^2$.  There are $q+1$ linear subspaces of dimension $1$ in $\F_q^2$.  The chosen generators are:
\begin{equation}  \label{eqn:ellVectors}
\ell_\infty = \langle 0, 1 \rangle, \ell_0 = \langle 1,0 \rangle, \text{ and } \ell_j = \langle 1, \al^j \rangle \text{ for } j=1, 2, \ldots, q-1.  
\end{equation}
Note that $\ell_{q-1} = \langle 1, 1 \rangle$.  Each of the generators $\ell_\mu$ contributes a block of representatives to $\OO$, as displayed in Figure~\ref{fig:OO}. 
\begin{figure}[b]
\caption{Representatives of $\sym(w)$-orbits in $\F_q^2$} \label{fig:OO}
\begin{align} 
\OO = \{ &\ell_\infty, \al \ell_\infty, \al^2 \ell_\infty, \ldots, \al^{t-1} \ell_\infty; \notag \\
&\ell_0, \al \ell_0, \al^2 \ell_0, \ldots, \al^{t-1} \ell_0; \notag \\
&\ell_1, \al \ell_1, \al^2 \ell_1, \ldots, \al^{t-1} \ell_1;  \label{eqn:OOReps} \\
&\vdots  \notag \\
&\ell_{q-1}, \al \ell_{q-1}, \al^2 \ell_{q-1}, \ldots, \al^{t-1} \ell_{q-1} \} . \notag
\end{align}
\end{figure}
The elements of $\OO$ also represent the $\sym(w)$-orbits on inputs.

A generator matrix $G$ for a linear code $C$ is determined by specifying how many times each element of $\OO$ appears as a column of $G$.  This information will be encoded in a \emph{multiplicity function} $\eta \colon  \OO \ra \N$.  That is, $\al^i \ell_\mu$ appears $\eta(\al^i \ell_\mu)$ times as a column of $G$.

Form a matrix $W$ of size $t(q+1) \times t(q+1)$, with rows and columns indexed by the elements of $\OO$:
\begin{equation}  \label{eqn:W-Matrix}
W_{\al^i \ell_\mu, \al^j \ell_\nu} = w(\al^i \ell_\mu \cdot \al^j \ell_\nu) , \quad \al^i \ell_\mu, \al^j \ell_\nu \in \OO,
\end{equation}
where $\cdot$ is the standard dot product on $\F_q^2$.
When $\eta$ is the multiplicity function that determines a generator matrix $G$, 
the matrix product $\omega = W \eta$ is a column vector with entries indexed by elements of $\OO$.  The entry indexed by $\al^i \ell_\mu$ is
\[  \omega_{\al^i \ell_\mu} = (W \eta)_{\al^i \ell_\mu} = \sum_{\al^j \ell_\nu \in \OO} w(\al^i \ell_\mu \cdot \al^j \ell_\nu) \eta(\al^j \ell_\nu) . \]
By \eqref{eqn:W-times-eta}, this expression is exactly $w(\al^i \ell_\mu G)$. In summary, the column vector $\omega = W \eta$ consists of the weights of all the codewords of $C$, organized by $\sym$-orbits.  We will refer to $\omega$ as the \emph{list of orbit-weights}.  Remembering that $w(uvG) = w(vG)$ for all $u \in \sym(w)$, the values of $\omega$, together with the sizes of the $\sym(w)$-orbits, determine the weight distribution of the linear code $C$:
\begin{equation}  \label{eqn:WeightDistributionFromOmega}
A_s(C) = \sum_{\substack{\al^i \ell_\mu \in \OO \colon  \\ \omega_{\al^i \ell_\mu}=s}} \size{[\al^i \ell_\mu]} ,
\end{equation}
where $[\al^i \ell_\mu]$ is the $\sym(w)$-orbit of $\al^i \ell_\mu$.

\section{Dual codewords of small weight}

Continue to assume that $R$ is a finite commutative ring with $1$.  Given an $R$-linear code $C \subseteq R^n$, denote by $\la_i$, $i=1, 2, \ldots, n$, the $i$th-coordinate projection $\la_i \colon  C \ra R$, where $\la_i(c_1, c_2, \ldots, c_n) = c_i$ for $c=(c_1, c_2, \ldots, c_n) \in C$.  Call $\la_i$ a \emph{coordinate functional};  $\la_i$ is a homomorphism of $R$-modules.
 
Define the \emph{effective length} of a linear code $C$ to be
\[ \efflength(C) = \size{\{ i  \colon  \la_i(c) \neq 0 \text{ for some $c \in C$}\} } . \] 
If $C$ is given by a generator matrix, the effective length counts the number of nonzero columns of the generator matrix.  The effective length ignores coordinate positions where every codeword is zero. 

A real-valued weight $w$ on $R$ is \emph{egalitarian} \cite{Heise-Honold:homogeneous-egalitarian} if there exists a positive constant $\zeta$ such that $\sum_{r \in I} w(r) = \zeta \size{I}$ for every nonzero ideal $I \subseteq R$. The term \emph{pre-homogeneous} is used in \cite{CHH}.

\begin{ex}  \label{ex:field-egalitarian}
Because a field contains only one nonzero ideal, any weight on a finite field with positive values is egalitarian.  The Lee weight on $\Z/2^a \Z$ is egalitarian with $\zeta = 2^{a-1}$.  The homogeneous weight on a finite Frobenius ring is egalitarian \cite[Corollary~1.6]{greferath-schmidt:combinatorics}.
\end{ex}

\begin{prop} \label{prop:egal-efflength}
Let $w$ be an egalitarian weight with constant $\zeta$ on a finite commutative ring $R$.  If $C \subseteq R^n$ is a linear code, then 
\[  \sum_{c \in C} w(c) = \zeta \size{C} \efflength(C) . \]
\end{prop}

\begin{proof}
Consider the coordinate functionals $\la_i \colon  C \ra R$, $i=1, 2, \ldots, n$. If $\la_i = 0$, then $i$ does not contribute to the effective length.  If $\la_i \neq 0$, then $i$ contributes to the effective length, and the image $\la_i(C)$ is a nonzero ideal of $R$.   

Every element $b \in \la_i(C)$ is hit $\size{\ker\la_i} = \size{C}/\size{\la_i(C)}$ times.  
Then
\begin{align*}
\sum_{c \in C} w(c) &= \sum_{c \in C} \sum_{i \colon  \la_i \neq 0} w(\la_i(c)) 
= \sum_{i \colon  \la_i \neq 0} \size{\ker \la_i} \sum_{b \in \la_i(C)} w(b)  \\
&= \sum_{i \colon  \la_i \neq 0} \frac{\size{C}}{\size{\la_i(C)}} \sum_{b \in \la_i(C)} w(b)  = \efflength(C) \size{C} \zeta  .  \qedhere
\end{align*}
\end{proof}

\begin{rem}
This result is essentially the same as a proposition in \cite[p.\ 319]{assmus-mattson:axiomatic}, which those authors credit to Slepian \cite[Proposition~6]{slepian:class}.
\end{rem}

\begin{cor}  \label{cor:effective-length}
Let $w$ be an egalitarian weight.  If $C_1, C_2 \subseteq R^n$ are two linear codes with the same $w$-weight distribution, then $C_1, C_2$ have the same effective length, $\efflength(C_1) = \efflength(C_2)$.
\end{cor}

\begin{proof}
The hypothesis is that $A_s(C_1) = A_s(C_2)$ for all $s$.  Because $\sum_s A_s(C) = \size{C}$, we have $\size{C_1} = \size{C_2}$.  As $\sum_s s A_s(C) = \sum_{c \in C} w(c)$, we also have $\sum_{c \in C_1} w(c) = \sum_{c \in C_2} w(c)$.  
Now apply Proposition~\ref{prop:egal-efflength}.
\end{proof}

Concentrate now on linear codes over a finite field $\F_q$.  By Example~\ref{ex:field-egalitarian}, every weight on $\F_q$ is egalitarian.  The next result is immediate from Corollary~\ref{cor:effective-length}.

\begin{cor}  \label{cor:effective-length-fields}
Let $w$ be any weight on $\F_q$.  If $C_1, C_2 \subseteq \F_q^n$ are linear codes with the same $w$-weight distribution, then $C_1, C_2$ have the same effective length, $\efflength(C_1) = \efflength(C_2)$.
\end{cor}

A vector $x \in \F_q^n$ will be called a \emph{singleton} if $x$ has exactly one nonzero entry.  A vector $x$ is called a \emph{doubleton} if it has exactly two nonzero entries.  Define 
\begin{align*}
\asing_j(C) &= \size{\{ x \in C \colon  \text{$x$ is a singleton and $w(x)=j$}\} } , \\
\adble_j(C) &= \size{\{ x \in C \colon  \text{$x$ is a doubleton and $w(x)=j$}\} } .
\end{align*}
Over $\F_q$, singleton dual codewords arise only from zero-columns in the original generator matrix.  This leads to the next result.

\begin{prop}  \label{prop:singletonsDoNotDistinguish}
Let $w$ be any weight on $\F_q$.  If $C_1, C_2 \subseteq \F_q^n$ are linear codes with the same $w$-weight distribution, then $\asing_s(C_1^\perp) = \asing_s(C_2^\perp)$ for all $s$.
\end{prop}

\begin{proof}
Any nonzero element $r \in \F_q$ is invertible.  If $r$ annihilates a coordinate functional $\la_i$, i.e., $\la_i r =0$, then $\la_i = 0$.  Thus, $\asing_s(C_1^\perp) = \asing_s(\F_q) (n - \efflength(C_1))$, where $\asing_s(\F_q)$ is the number of elements of $\F_q$ of weight $s$.  The result now follows from Corollary~\ref{cor:effective-length-fields}.
\end{proof}

For doubletons, we assume a linear code is presented via a multiplicity function $\eta$, as in Section~\ref{sec:MultFcns}.  We also assume $w$ has the properties that $-1 \in \sym(w)$ and that the minimum positive value $\mathring{w}$ of $w$ is taken on a unique $\sym(w)$-orbit.

\begin{lem}  \label{lem:2MathringInequality}
Let $w$ be a weight on a finite ring $R$ with $w(0)=0$ and $w(r)>0$ for nonzero $r \in R$.  Let $\mathring{w}$ be the smallest positive value of $w$ on $R$.  If $v \in R^n$ satisfies $\mathring{w} \leq w(v) < 2 \mathring{w}$, then $v$ is a singleton.  If $w(v) = 2 \mathring{w}$, then $v$ is either a singleton \textup{(}whose nonzero entry has weight $2 \mathring{w}$\textup{)} or a doubleton \textup{(}both of whose nonzero entries have weight $\mathring{w}$\textup{)}.  In particular, for any linear code $C \subseteq R^n$, 
\[  A_{2 \mathring{w}}(C) = \asing_{2 \mathring{w}}(C) + \adble_{2 \mathring{w}}(C) . \]
\end{lem}

\begin{proof}
If a vector $v \in R^n$ has $n'$ nonzero entries, then $w(v) \geq n' \mathring{w}$.  Equality occurs if and only if each nonzero entry has weight $\mathring{w}$.
\end{proof}

\begin{prop}  \label{prop:2w-weights}
Let $w$ be a weight on $\F_q$ with $-1 \in \sym(w)$ and $\mathring{w}$ achieved on a unique $\sym(w)$-orbit.  Let $C$ be a linear code over $\F_q$ with no zero-positions given by a multiplicity function $\eta$.  Then
\[  A_{2 \mathring{w}}(C^\perp) = \adble_{2 \mathring{w}}(C^\perp) = \size{\sym(w)} \sum_{\iota} \binom{\eta_\iota}{2} . \]
\end{prop}

\begin{proof}
Because there are no zero-positions in $C$, $\asing_s(C^\perp) = 0$ for all $s>0$.  Then $A_{2 \mathring{w}}(C^\perp) = \adble_{2 \mathring{w}}(C^\perp)$, by Lemma~\ref{lem:2MathringInequality}.

Let $x \in C^\perp$ be a doubleton with $w(x) = 2 \mathring{w}$.  Being a doubleton, $x$ has exactly two nonzero entries, say, $x_i, x_j$.  Because $2 \mathring{w} = w(x) = w(x_i) + w(x_j)$ and each of $w(x_i), w(x_j)$ is at least  $\mathring{w}$, we see that $w(x_i) = w(x_j) = \mathring{w}$.  By hypothesis, $x_i, x_j$ must belong to the unique $\sym(w)$-orbit that minimizes $w$.  By the hypothesis $-1 \in \sym(w)$, we see that $-x_i$ belongs to the same $\sym(w)$-orbit as $x_i, x_j$.  

Let the minimizing orbit be $[\mathring{r}]$.  Then $x_i = h_i \mathring{r}$ and $x_j =  h_j \mathring{r}$ for some $h_i, h_j \in \sym(w)$.  Then $-x_i x_j^{-1} = -h_i h_j^{-1} \in \sym(w)$.

Because $x \in C^\perp$, we have $\la_i x_i + \la_j x_j = 0$.  In particular, $\la_i, \la_j$ are linearly dependent, and $\la_j = - \la_i x_i x_j ^{-1} = \la_i (-h_i h_j^{-1})$.  In particular, $\la_j , \la_i$ are in the same $\sym(w)$-orbit.  By unique choice of representatives in $\OO$, we have $\la_j = \la_i$.  This, in turn implies $x_j = - x_i$.

We see that contributions to $\adble_{2 \mathring{w}}(C^\perp)$ occur by choosing two coordinate functionals from the same $\sym(w)$-orbit and choosing an element $h \in \sym(w)$.  Then set $x_i = h \mathring{r}$ and $x_j = -h \mathring{r}$.
\end{proof}

\begin{rem}
If there are $n_0 = n - \efflength(C)$ zero-positions, then there are an additional $\size{\sym(w)}^2 \binom{n_0}{2}$ doubletons of weight $2 \mathring{w}$ in $C^\perp$.  There is no longer any relation between $x_i$ and $x_j$.
\end{rem}

\section{Some structured matrices}
Given a weight $w$ on $\F_q$, the matrix $W$ defined in \eqref{eqn:W-Matrix} turns out to have considerable structure.
In this section we examine this and a related structure in some detail.

Fix a positive integer $n$.  Let $I_n$ be the $n \times n$ identity matrix and $J_n$ be the $n \times n$ all-one matrix; i.e., every entry of $J_n$ equals $1$.  For one type of structure, consider all complex matrices of the form
\begin{equation} \label{eqn:simpleStructured}
M_{x,y} = (x-y) I_n + y J_n = \begin{bmatrix}
x & y & y &\cdots & y \\
y & x & y & \cdots & y \\
\vdots & & \ddots &  & \vdots \\
y & y & \cdots & x & y \\
y & y & \cdots & y & x
\end{bmatrix} . 
\end{equation}

\begin{lem}
If $M_{x,y}$ has the form in \eqref{eqn:simpleStructured}, then
\[  \det M_{x,y} = (x+(n-1)y)(x-y)^{n-1} . \]
\end{lem}

\begin{proof}
Subtract the first row from every other row.  Then add columns $2$ through $n$ to the first column.  These row and column operations do not change the value of the determinant.  The result of the row and column operations is an upper triangular matrix with determinant as claimed.
\end{proof}

\begin{lem}  \label{lem:preserveInverseType1}
Suppose $M_{x,y}$ is as in \eqref{eqn:simpleStructured}.  If $M_{x,y}$ is invertible, then $M_{x,y}^{-1}$ also has the form in \eqref{eqn:simpleStructured}.
\end{lem}

\begin{proof}
Multiply two matrices of the form given in \eqref{eqn:simpleStructured}: $M_{x,y} M_{z,w}$.
The product has the form $M_{\delta, \epsilon}$, with
\begin{align*}
\delta &= xz+(n-1)yw, \\
\epsilon &= xw+yz+(n-2)yw .
\end{align*}
Set $\delta=1, \epsilon=0$, and solve the resulting system of equations for $z,w$:
\[  \begin{bmatrix}
x & (n-1)y \\ y & x + (n-2)y
\end{bmatrix}
\begin{bmatrix}
z \\ w
\end{bmatrix} =
\begin{bmatrix}
1 \\ 0
\end{bmatrix}. \]
A unique solution exists, because the coefficient matrix is invertible.  Indeed, its determinant, $(x-y)(x+(n-1)y)$, has the same factors as $\det M_{x,y}$, which is nonzero by hypothesis.
\end{proof}

A second type of structure comes from matrices that look like \eqref{eqn:W-Matrix}.  It will be fruitful to generalize the discussion slightly.

Fix a subgroup $H \subsetneq \F_q^\times$.  (In the context of \eqref{eqn:W-Matrix}, $H = \sym(w)$ for a weight $w$ on $\F_q$.)  
Fix a primitive element $\al$ of $\F_q^\times$, and let $t$ be the smallest positive integer so that $\al^t$ generates $H$; then $t \size{H} = \size{\F_q^\times} = q-1$.  As $H \neq \F_q^\times$, $t \geq 2$.  We say that a function $a \colon  \F_q \ra \C$ is \emph{$H$-invariant} if $a(hr) = a(r)$ for all $r \in \F_q$ and all $h \in H$.  By $H$-invariance, every value of $a$ is one of $a(0), a(1), a(\al), \ldots, a(\al^{t-1})$.  (In displays, we sometimes write $a_r$ instead of $a(r)$.)

Generalizing \eqref{eq:R-W}, we associate to every $H$-invariant function $a \colon  \F_q \ra \C$ a $t \times t$ circulant matrix $\mathcal{A}$ whose rows and columns are indexed by $\{0, 1, \ldots, t-1\}$:
\begin{equation}  \label{eqn:CirculantA}
\mathcal{A}_{i,j} = a(\al^{i+j}) , \quad i,j \in \{0, 1, \ldots, t-1\} .
\end{equation}

Recall the vectors $\ell_\mu$, $\mu \in \{\infty, 0, 1, \ldots, q-1\}$, from \eqref{eqn:ellVectors}, and $\OO$ from \eqref{eqn:OOReps}, using nonzero $H$-orbits.  The dot product of $x, y \in \F_q^2$ is denoted by $x \cdot y$.

\begin{lem}  \label{lem:one-dot}
For the vector $\ell_\mu$, $\mu \in \{\infty, 0, 1, \ldots, q-1\}$, define $f \colon  \F_q^2 \ra \F_q$ by $f(x) = x \cdot \ell_\mu$.  Then there is a unique $\nu \in \{\infty, 0, 1, \dots, q-1\}$ such that $f(\ell_\nu) = 0$.
\end{lem}

\begin{proof}
The dot product on $\F_q^2$ is nondegenerate, so $f$ is nonzero.  Being linear,  $f$ is surjective and hence $\dim \ker f = 1$.  As any two distinct $\ell_\nu, \ell_\sigma$ are linearly independant, at most one $\ell_\nu$ can be contained in $\ker f$.  If $\mu = 0$ or $\infty$, use $\nu = \infty$ or $0$, respectively.  For $\mu \neq 0, \infty$, use $\nu$ such that $\al^\nu = -\al^{-\mu}$.
\end{proof}

Write $\ell_{\mu}^\perp$ for $\ell_{\nu}$ when $\ell_{\nu} \cdot \ell_{\mu} = 0$.  The next lemma describes what happens to $H$-orbits under certain linear transformations $\F_q^2 \ra \F_q^2$.

\begin{lem}  \label{lem:two-dot}
For $\al^i \ell_\mu, \al^j \ell_\nu \in \OO$, define $F \colon  \F_q^2 \ra \F_q^2$ by $F(x) = \langle x \cdot \al^i \ell_\mu, x \cdot \al^j \ell_\nu \rangle$.  As $x = \al^k \ell_\sigma$ varies over $\OO$,   
\begin{itemize}
\item when $\mu=\nu$: $F$ sends the $t$ orbits contained in $\ell_\mu^\perp$ to $0$, and hits each of the $t$ orbits contained in $\langle 1, \al^{j-i} \rangle$ a total of $q$ times;
\item  when $\mu \neq \nu$: $F$ hits every nonzero $H$-orbit in $\F_q^2$ exactly once. 
\end{itemize}
\end{lem}

\begin{proof}
When $\mu \neq \nu$, the vectors $\al^i \ell_\mu, \al^j \ell_\nu$ are linearly independant, so $F$ is an isomorphism.  The elements of $\OO$ are representatives of the (nonzero) $H$-orbits, and the $H$-orbits partition $\F_q^2$.  An isomorphism maps the orbit partition back to itself.  So every $H$-orbit must be hit exactly once.

When $\mu=\nu$, the vectors $\al^i \ell_\mu, \al^j \ell_\nu$ are linearly dependant, so $F$ has rank $1$, with image equal to the $1$-dimensional subspace of $\F_q^2$ spanned by $\langle \al^i, \al^j \rangle$.  Indeed, $F(x) = (x \cdot \ell_\mu) \langle \al^i, \al^j \rangle = (x \cdot \ell_\mu) \al^i \langle 1, \al^{j-i} \rangle$.  As we saw in Lemma~\ref{lem:one-dot}, $t$ nonzero $H$-orbits map to zero (those contained in $\ell_\mu^\perp$).  The remaining $qt$ nonzero $H$-orbits cover, $q$ times each, the $t$ orbits contained in  $\langle \al^i, \al^j \rangle$.
\end{proof}

Here is the definition of the second type of structure, which is a slight generalization of \eqref{eqn:W-Matrix}, because $a(0)$ is not restricted to be equal to $0$.  For every $H$-invariant function $a \colon  \F_q \ra \C$, define a $t(q+1) \times t(q+1)$ matrix $A$, whose rows and columns are indexed by elements of $\OO$:
\begin{equation}  \label{eqn:StrMat2}
A_{\al^i \ell_\mu, \al^j \ell_\nu} = a(\al^i \ell_\mu \cdot \al^j \ell_\nu) .
\end{equation}
The indices satisfy $i,j = 0, 1, \ldots, t-1$ and $\mu, \nu \in \{\infty, 0, 1, \ldots, q-1\}$.  Note that $A$ is a symmetric matrix.  The matrix $A$ has a block form inherited from the blocks of $\OO$.

\begin{ex}  \label{ex:F5Amatrix}
In $\F_5$, let $H = \{ \pm 1\} \subseteq \F_5^\times$.  Take $\al = 2$.  Then $t=2$ and
\begin{align*}
\OO = \{ \langle0,1\rangle, &\langle0,2\rangle, \langle1,0\rangle, \langle2,0\rangle, 
 \langle1,2\rangle, \langle2,4\rangle, \\
 & \langle1,4\rangle, \langle2,3\rangle , \langle1,3\rangle, \langle2,1\rangle, \langle1,1\rangle, \langle2,2\rangle  \} .
\end{align*}
Let $a \colon  \F_5 \ra \C$ be $H$-invariant.  Then $a_4=a_1$, $a_3=a_2$, and 
\[  A=\left[\begin{array}{rr|rr|rr|rr|rr|rr}
a_{1} & a_{2} & a_{0} & a_{0} & a_{2} & a_{1} & a_{1} & a_{2} & a_{2} & a_{1} & a_{1} & a_{2} \\
a_{2} & a_{1} & a_{0} & a_{0} & a_{1} & a_{2} & a_{2} & a_{1} & a_{1} & a_{2} & a_{2} & a_{1} \\ \hline
a_{0} & a_{0} & a_{1} & a_{2} & a_{1} & a_{2} & a_{1} & a_{2} & a_{1} & a_{2} & a_{1} & a_{2} \\
a_{0} & a_{0} & a_{2} & a_{1} & a_{2} & a_{1} & a_{2} & a_{1} & a_{2} & a_{1} & a_{2} & a_{1} \\ \hline
a_{2} & a_{1} & a_{1} & a_{2} & a_{0} & a_{0} & a_{1} & a_{2} & a_{2} & a_{1} & a_{2} & a_{1} \\
a_{1} & a_{2} & a_{2} & a_{1} & a_{0} & a_{0} & a_{2} & a_{1} & a_{1} & a_{2} & a_{1} & a_{2} \\ \hline
a_{1} & a_{2} & a_{1} & a_{2} & a_{1} & a_{2} & a_{2} & a_{1} & a_{2} & a_{1} & a_{0} & a_{0} \\
a_{2} & a_{1} & a_{2} & a_{1} & a_{2} & a_{1} & a_{1} & a_{2} & a_{1} & a_{2} & a_{0} & a_{0} \\ \hline
a_{2} & a_{1} & a_{1} & a_{2} & a_{2} & a_{1} & a_{2} & a_{1} & a_{0} & a_{0} & a_{1} & a_{2} \\
a_{1} & a_{2} & a_{2} & a_{1} & a_{1} & a_{2} & a_{1} & a_{2} & a_{0} & a_{0} & a_{2} & a_{1} \\ \hline
a_{1} & a_{2} & a_{1} & a_{2} & a_{2} & a_{1} & a_{0} & a_{0} & a_{1} & a_{2} & a_{2} & a_{1} \\
a_{2} & a_{1} & a_{2} & a_{1} & a_{1} & a_{2} & a_{0} & a_{0} & a_{2} & a_{1} & a_{1} & a_{2}
\end{array}\right]. \]
\end{ex}

Example~\ref{ex:F5Amatrix} illustrates the following general result.

\begin{lem}  \label{lem:blockTypes}
For any $H$-invariant function $a \colon  \F_q \ra \C$, the blocks of the matrix $A$ are of two types:
\begin{itemize}
\item the block is $a(0)$ times the all-$1$ matrix;  or
\item the block is circulant: $a(1), a(\al), \ldots, a(\al^{t-1})$, or a cyclic shift thereof.
\end{itemize}
For a given block-row, there is one block of the first type and $q$ blocks of the second type.  For every row of $A$, there are $t$ entries equal to $a(0)$ and $q$ entries of each of $a(1), a(\al), \ldots, a(\al^{t-1})$.  The sum of each row's entries is $t a(0) + q \sum_{j=0}^{t-1} a(\al^j)$.
\end{lem}

\begin{proof}
The block determined by $\ell_\mu$ and $\ell_\nu$ has $i,j$-entry equal to $a(\al^i \ell_\mu \cdot \al^j \ell_\nu) = a(\al^{i+j} (\ell_\mu \cdot \ell_\nu))$.  If $\ell_\mu \cdot \ell_\nu = 0$, then all the entries for that block equal $a(0)$.  That is, the block is $a(0)$ times the all-$1$ matrix of size $t \times t$.

If $\ell_\mu \cdot \ell_\nu \neq 0$, then $\ell_\mu \cdot \ell_\nu = \al^k$ for some $k \in \{0, 1, \ldots, q-2 \}$.  Then the $i,j$-entry equals $a(\al^{i+j+k})$.  By the $H$-invariance of $a$, this yields a circulant matrix beginning with $a(\al^k)$.

For each $\ell_\mu$, there is a unique $\ell_\nu$ with $\ell_\mu \cdot \ell_\nu = 0$, by Lemma~\ref{lem:one-dot}. 
\end{proof}

\begin{cor}  \label{cor:rowSumsOfW}
Suppose $a(0)=0$.  Then the sum of the entries of any row or column of $A$ is $q \sum_{j=0}^{t-1} a(\al^j)$.  If $A$ is invertible, then the sum of the entries in any row or column of $A^{-1}$ is $1/(q \sum_{j=0}^{t-1} a(\al^j))$.
\end{cor}

\begin{proof}
The row sum comes from Lemma~\ref{lem:blockTypes} and $a(0)=0$.  This means the all-one vector $\mathbf{1}$ is an eigenvector for $A$ with eigenvalue $q \sum_{j=0}^{t-1} a(\al^j)$.  If $A$ is invertible, multiplying the eigenvector equation by $A^{-1}$ shows that $\mathbf{1}$ is an eigenvector for $A^{-1}$ with eigenvalue $1/(q \sum_{j=0}^{t-1} w(\al^j))$.  The column sums follow from $A$ being a symmetric matrix.
\end{proof}

Now consider a second $H$-invariant function $b \colon  \F_q \ra \C$ with associated matrix $B$ of size $t(q+1) \times t(q+1)$.  As above, $H$-invariance implies that every value of $b$ equals one of $b(0), b(1), b(\al), b(\al^2), \ldots, b(\al^{t-1})$.

For $H$-invariant functions $a, b \colon  \F_q \ra \C$, define $\breve{a}$ and the \emph{correlation} $c_m(a,b)$ with shift $m$, $m=0, 1, \ldots, t-1$, by
\begin{equation}  \label{eq:correlation}
\breve{a} = \sum_{j=0}^{t-1} a(\al^j), \quad c_m(a,b) = \sum_{j=0}^{t-1} a(\al^j) b(\al^m \al^j) .
\end{equation}

\begin{lem}
The correlations $c_m(a,b)$ of $H$-invariant functions satisfy two properties:
\begin{itemize}
\item  $c_{m+t}(a,b) = c_m(a,b)$, so the index $m$ makes sense modulo $t$;
\item  $c_{t-m}(a,b) = c_m(b,a)$; in particular, $c_{t-m}(a,a) = c_m(a,a)$.
\end{itemize}
\end{lem}

\begin{proof}
The first property is a consequence of $H$-invariance.  The second follows by reindexing the sums.
\end{proof}

We now discuss the product $AB$ in terms of blocks.
\begin{lem}  \label{lem:ABentries}
The product $AB$ has the following features:
\begin{itemize}
\item The diagonal blocks are correlation matrices of the form
\[  \begin{bmatrix}
t a_0 b_0 + q c_0 & t a_0 b_0 + q c_1 & \cdots & t a_0 b_0 +q c_{t-1} \\
t a_0 b_0 + q c_{t-1} & t a_0 b_0 + q c_0 & \cdots & t a_0 b_0 + q c_{t-2} \\
\vdots & \vdots & \ddots & \vdots \\
t a_0 b_0 + q c_1 & t a_0 b_0 + q c_2 & \cdots & t a_0 b_0 + q c_0
\end{bmatrix}, \]
where $c_m=c_m(a,b)$ throughout.  
\item  All the off-diagonal blocks are the same, namely, all-$1$ matrices multiplied by
$a_0  \breve{b} + \breve{a} b_0 + ((q-1)/t) \breve{a} \breve{b}$.
\end{itemize}
\end{lem}

\begin{proof}
Begin with the off-diagonal block case, at position indexed by $\al^i \ell_\mu, \al^j \ell_\nu$, with $\mu \neq \nu$.  The corresponding entry $E_o$ of $AB$ is
\[  E_o=\sum_{\al^k \ell_{\sigma} \in \OO} a(\al^{i+k} \ell_\mu \cdot \ell_{\sigma}) \, b(\al^{k+j} \ell_{\sigma} \cdot \ell_\nu) . \]
As we saw in Lemma~\ref{lem:two-dot}, every $H$-orbit is hit once.  Using the $H$-invariance of $a$ and $b$, we have
\begin{align*}
E_o&= \sum_{k=0}^{t-1} a(0) b(\al^k) + \sum_{k=0}^{t-1} \sum_{r \in \F_q} a(\al^k) b(\al^k r) \\
&= \sum_{k=0}^{t-1} a(0) b(\al^k) + \sum_{k=0}^{t-1} a(\al^k) b(0) +\sum_{k=0}^{t-1} \sum_{m=0}^{q-2} a(\al^k) b(\al^{k + m}) \\
&= \sum_{k=0}^{t-1} a(0) b(\al^k) + \sum_{k=0}^{t-1} a(\al^k) b(0) + \frac{q-1}{t} \sum_{k=0}^{t-1} \sum_{m=0}^{t-1} a(\al^k) b(\al^{k + m}) \\
&= \sum_{k=0}^{t-1} a(0) b(\al^k) + \sum_{k=0}^{t-1} a(\al^k) b(0) +\frac{q-1}{t} \sum_{k=0}^{t-1}  a(\al^k) \sum_{m'=0}^{t-1} b(\al^{m'}) .
\end{align*}
In the last step, we re-indexed via $m'=k + m$.

Now look at a diagonal block, so that $\mu=\nu$.  The entry $E_d$ at position indexed by $\al^i \ell_\mu, \al^j \ell_\mu$ is 
\[ E_d = \sum_{\al^k \ell_{\sigma} \in \OO} a(\al^{i + k} \ell_\mu \cdot \ell_{\sigma}) \, b(\al^{k + j} \ell_{\sigma} \cdot \ell_\mu) . \]
Using Lemma~\ref{lem:one-dot}, 
there is one $\ell_\tau$ such that $\ell_\mu \cdot \ell_\tau=0$.  For the remaining $q$ vectors $\ell_{\sigma}$, $\ell_\mu \cdot \ell_{\sigma} \neq 0$, so $\ell_\mu \cdot \ell_{\sigma} = \al^{m_{\sigma}}$, for some $m_{\sigma} = 0, 1, \ldots, q-2$.  Isolate the contribution to $E_d$ of the orbits contained in one $\ell_{\sigma}$:
\begin{align*}
\sum_{k=0}^{t-1} a(\al^{i + k} &\ell_\mu \cdot \ell_{\sigma}) \, b(\al^{k + j} \ell_{\sigma} \cdot \ell_\mu) \\
&= \begin{cases}
t a(0) b(0), & \ell_\mu \cdot \ell_\sigma = 0, \\
\sum_{k=0}^{t-1} a(\al^{i + k} \al^{m_{\sigma}}) \, b(\al^{k + j} \al^{m_{\sigma}}) , & \ell_\mu \cdot \ell_{\sigma} \neq 0,
\end{cases} \\
&= \begin{cases}
t a(0) b(0), & \ell_\mu \cdot \ell_{\sigma} = 0, \\
\sum_{k=0}^{t-1} a(\al^{i + k}) \, b(\al^{k + j}) , & \ell_\mu \cdot \ell_{\sigma} \neq 0,
\end{cases}
\end{align*}
using $H$-invariance and reindexing the summation.
Adding the contributions of all $\ell_{\sigma}$ yields
\begin{align*}
E_d &= t a(0) b(0) + q \sum_{k=0}^{t-1} a(\al^{i + k}) \, b(\al^{k + j}) \\
&=  t a(0) b(0) + q c_{j - i}(a,b) .  \qedhere
\end{align*}
\end{proof}

\begin{cor}  \label{cor:solveEquations}
Suppose $a \colon  \F_q \ra \C$ is $H$-invariant and that $a(0)=0$.  Assume the matrix $\mathcal{A}$ of \eqref{eqn:CirculantA} is invertible.  Then the matrix $A$ of \eqref{eqn:StrMat2} is also invertible, and the inverse matrix $A^{-1}$ has the form of \eqref{eqn:StrMat2}.
\end{cor}

\begin{proof}
Here is a system of linear equations for an unknown $H$-invariant function $b$:
\begin{align*}
c_0(b,a) &= 1/q, \\
c_1(b,a) &=0, \\
& \vdots \\
c_{t-1}(b,a) &= 0, \\
b(0) &= -((q-1)/t) \left( b(1) +\cdots + b(\al^{t-1}) \right) .
\end{align*}
The coefficient matrix for the first $t$ equations is exactly the matrix $\mathcal{A}$, which is invertible by hypothesis. Then the first $t$ equations have a unique solution $b(1), b(\al), \ldots, b(\al^{t-1})$, and the last equation gives $b(0)$.   By Lemma~\ref{lem:ABentries}, $BA=I_{t (q+1)}$,  so $B = A^{-1}$.
\end{proof}

\begin{rem}
The definitions of structured matrices can be extended beyond \eqref{eqn:StrMat2}, where $H$-orbits in $\F_q^2$ were used to define \eqref{eqn:OOReps}.  There are similar results using $H$-orbits in $\F_q^k$ for any integer $k \geq 2$.
\end{rem}

\section{Constructing two-weight codes}
In this section we describe the construction of several two-weight codes, all of which have the same weight distribution.  This is accomplished by solving the matrix equation $W \eta = \omega$ using different choices of $\omega$.

Assume $w$ is a nondegenerate weight on $\F_q$, with $w(0)=0$.  This means that the matrix $\mathcal{W}$ of \eqref{eqn:CirculantA} is invertible.  As usual, we will assume that $w$ takes positive integer values, except for $w(0)=0$.  Write $H = \sym(w) \subseteq \F_q^\times$, and let $t$ be the smallest positive integer such that $\al^t$ generates $H$, where $\al$ is a primitive element of $\F_q$.  Assume that $H \neq \F_q^\times$, so that $t\size{H} = (q-1)$, with $t \geq 2$.  By $H$-invariance, the values $w(1), w(\al), w(\al^2), \ldots, w(\al^{t-1})$ completely determine $w$.  

The linear codes to be constructed will have dimension $2$.   The set of nonzero $H$-orbits in $\F_q^2$ is represented by $\OO$ of \eqref{eqn:OOReps}.   Write the number of such orbits as $\tau=\size{\OO} = t(q+1)$.  The matrix $W$ of  \eqref{eqn:W-Matrix} is invertible by Corollary~\ref{cor:solveEquations}.  For vectors $\eta, \omega \in \Q^\tau$, we will examine carefully solutions of the equation $W \eta = \omega$ when $\omega$ assumes only two positive values.  We are interested in understanding when there are solutions of $W \eta = \omega$ with $\eta$ having all values nonnegative.

View $\Q^\tau$ as $\Q^\OO$, the space of functions from $\OO$ to $\Q$, and write $\la \in \OO$ for a typical element of $\OO$.  
Let $\varrho \in \Q$ be positive, $\varrho \neq 1$, and choose two subsets $S, S' \subseteq \OO$ of the same size: $s=\size{S} = \size{S'}$.  Define elements 
$\omega, \omega' \in \Q^\OO$ by
\begin{equation}  \label{eqn:defineOmegas}
 \omega(\la) = \begin{cases}
\varrho, & \la \in S, \\
1, & \text{otherwise} ;
\end{cases}  \quad
\omega'(\la) = \begin{cases}
\varrho, & \la \in S', \\
1, & \text{otherwise} .
\end{cases} 
\end{equation}

\begin{lem}  \label{lem:sameLength}
Let $\omega, \omega'$ be as in \eqref{eqn:defineOmegas}, with $s=\size{S} = \size{S'}$.
The solutions $\eta, \eta'$ of the equations $W \eta = \omega$, $W \eta' = \omega'$ satisfy $\sum_{\la \in \OO} \eta(\la) = \sum_{\la \in \OO} \eta'(\la)$.
\end{lem}

\begin{proof}
Left multiply $\omega = W \eta$ by the all-$1$ row vector $\mathbf{1}$.  By Corollary~\ref{cor:rowSumsOfW}, $\mathbf{1}$ is a left eigenvector of $W$, so that 
\begin{align*}
\mathbf{1} \omega &=  \mathbf{1} W \eta ,\\
s \varrho + \tau - s &= q (\sum_{j=0}^{t-1} w(\al^j)) \mathbf{1} \eta = q (\sum_{j=0}^{t-1} w(\al^j)) \sum_{\la \in \OO} \eta(\la) .
\end{align*}
Doing the same for $\omega' = W \eta'$ yields the result.
\end{proof}

The following are the key technical results of the paper.  The reader may find it useful to refer to Example~\ref{ex:F5} while reading the proofs.

\begin{lem}  \label{lem:NonNegEtas}
Suppose $w$ is a nondegenerate rational-valued weight on $\F_q$, with $w(0)=0$, $w(r) >0$ for $r \neq 0$, and $H = \sym(w) \neq \F_q^\times$. 

Fix an integer $s$ satisfying $0 < s < \tau = \size{\OO}$.  Then there exist rational numbers $\delta_{\min}^{(s)}, \delta_{\max}^{(s)}$, with $0 < \delta_{\min}^{(s)} < 1 < \delta_{\max}^{(s)}$, such that:
for any subset $S \subset \OO$, with $\size{S} =s$, and any $\varrho$ with $\delta_{\min}^{(s)} \leq \varrho \leq \delta_{\max}^{(s)}$, the solution $\eta$ of the equation $W \eta = \omega$, with $\omega$ as in \eqref{eqn:defineOmegas}, has $\eta(\la) \geq 0$ for all $\la \in \OO$.  In particular, the numbers $\delta_{\min}^{(s)}, \delta_{\max}^{(s)}$ do not depend on the choice of subset $S$ of size $s$.
\end{lem}

\begin{proof}
Take any subset $S \subset \OO$ with $\size{S}=s$, and consider $\omega$ from \eqref{eqn:defineOmegas} with an unknown value of $\varrho$.  Because $w$ is nondegenerate, $W$ is invertible.   Thus $\eta = W^{-1} \omega$ exists.  This implies that $\eta(\la) = P(\la) + Q(\la) \varrho$, where $Q(\la)$ is the sum of the entries of $W^{-1}$ in row $\la$ and columns indexed by elements of $S$, and $P(\la)$ is the sum of the remaining entries in row $\la$.  By Corollary~\ref{cor:rowSumsOfW}, $P(\la) + Q(\la) >0$, for all $\la \in \OO$.  This implies, for each $\la$, that at most one of $P(\la), Q(\la)$ can be negative.  

If $P(\la) <0$, then $Q(\la) > \abs{P(\la)} > 0$, and $\eta(\la) \geq 0$ if and only if $\varrho \geq \abs{P(\la)}/Q(\la)$.  Observe that $0 < \abs{P(\la)}/Q(\la) < 1$.  Similarly, if $Q(\la) < 0$, then $P(\la) > \abs{Q(\la)} >0$, and $\eta(\la) \geq 0$ if and only if $\varrho \leq P(\la)/\abs{Q(\la)}$.  Note that $P(\la)/\abs{Q(\la)} > 1$.  There is no extra restriction on $\varrho$ when both $P(\la)$ and $Q(\la)$ are nonnegative.

Define $\delta_{\min}^S, \delta_{\max}^S$ as follows:
\begin{align*}
\delta_{\min}^S &= \max\{ \abs{P(\la)}/Q(\la) \colon  P(\la) < 0\}  , \\
\delta_{\max}^S &= \min\{ P(\la)/\abs{Q(\la)} \colon  Q(\la) < 0\}  .
\end{align*}
Because we are finding maximums amd minimums of finite sets of positive rational numbers, we see that $0 < \delta_{\min}^S < 1 < \delta_{\max}^S$.  If $\varrho$ satisfies $\delta_{\min}^S \leq \varrho \leq \delta_{\max}^S$, then $\varrho$ satisfies all the inequalities in the previous paragraph, and every $\eta(\la)$ is nonnegative.

Now optimize over subsets $S \subset \OO$ with $\size{S}=s$.  Define:
\begin{align*}
\delta_{\min}^{(s)} &= \max\{\delta_{\min}^S \colon  S \subset \OO, \size{S}=s\}  , \\
\delta_{\max}^{(s)} &= \min\{ \delta_{\max}^S \colon  S \subset \OO, \size{S}=s\}  .
\end{align*}
As there are only finitely many subsets of $\OO$ of size $s$, we still have $0 < \delta_{\min}^{(s)} < 1 < \delta_{\max}^{(s)}$.  If $\varrho$ satisfies $\delta_{\min}^{(s)} \leq \varrho \leq \delta_{\max}^{(s)}$, then $\delta_{\min}^S \leq \varrho \leq \delta_{\max}^S$ for all $S$ with $\size{S}=s$, so that $\eta(\la) \geq 0$, $\la \in \OO$, for any choice of $S$ with $\size{S} = s$.
\end{proof}

Define $\mathfrak{q}(\eta) = \sum_{\la \in \OO} \eta(\la)(\eta(\la)-1)$, $\eta \in \Q^\OO$.

\begin{thm}  \label{thm:differentQs}
Suppose $w$ is a nondegenerate rational-valued weight on $\F_q$, with $w(0)=0$, $w(r) >0$ for $r \neq 0$, and $H = \sym(w) \neq \F_q^\times$.  
Assume the values $((q-1)/t) \breve{w}^2, q c_m(w,w)$, $m=1, 2, \ldots, \lfloor t/2 \rfloor$, of \eqref{eq:correlation}, are not all equal.
Then there exist two subsets $S, S' \subseteq \OO$, with $\size{S} = \size{S'} =2$, such that, for all but at most two values of $\varrho$ satisfying $\delta_{\min}^{(2)} \leq \varrho \leq \delta_{\max}^{(2)}$ and for $\omega, \omega'$ as in \eqref{eqn:defineOmegas}, the solutions $\eta, \eta'$ of the equations $W \eta = \omega, W \eta' = \omega'$ have $\eta(\la) \geq 0, \eta'(\la) \geq 0$ for all $\la \in \OO$, as well as $\mathfrak{q}(\eta) \neq \mathfrak{q}(\eta')$.
\end{thm} 

\begin{proof}
We have already seen that the hypothesis that $w$ is nondegenerate means that the matrix $\mathcal{W}$ of \eqref{eqn:CirculantA} is invertible.  Corollary~\ref{cor:solveEquations} then implies that the matrix $W$ of \eqref{eqn:StrMat2} is invertible and that its inverse $W^{-1}$ also has form \eqref{eqn:StrMat2}.  Let $b \colon  \F_q \ra \Q$ be the $H$-invariant function given by Corollary~\ref{cor:solveEquations} with $B = W^{-1}$.

The hypothesis that $((q-1)/t) \breve{w}^2, q c_m(w,w)$, $m=1, 2, \ldots, \lfloor t/2 \rfloor$, are not all equal, and Lemma~\ref{lem:ABentries}, imply that $W^2$ does not have form \eqref{eqn:simpleStructured}.  Then, Lemma~\ref{lem:preserveInverseType1} implies that $B^2=W^{-2}$ also does not have form \eqref{eqn:simpleStructured}.  Thus, Lemma~\ref{lem:ABentries}, applied to $B^2=W^{-2}$, implies that $((q-1)/t) \breve{b}^2, t b_0^2+q c_m(b,b)$, $m=1, 2, \ldots, \lfloor t/2 \rfloor$, are not all equal.

There are two cases: 
\begin{itemize}
\item $((q-1)/t) \breve{b}^2 \neq t b_0^2 + q c_m(b,b)$ for some $m$, or
\item $t b_0^2 + q c_{m_1}(b,b) \neq t b_0^2 + q c_{m_2}(b,b)$ for some $m_1 \neq m_2$.
\end{itemize}

If $((q-1)/t) \breve{b}^2 \neq t b_0^2 + q c_m(b,b)$, set $S = \{ \ell_0, \al^m \ell_0\}$ and $S'=\{ \ell_0, \ell_\infty \}$.   If $t b_0^2 + q c_{m_1}(b,b) \neq t b_0^2 + q c_{m_2}(b,b)$, set $S = \{ \ell_0, \al^{m_1} \ell_0\}$ and $S'=\{ \ell_0, \al^{m_2} \ell_0 \}$.  In both cases, $\size{S} = \size{S'} =2$.

Applying Lemma~\ref{lem:NonNegEtas} with $s=2$ to both $S$ and $S'$, we see that $\eta(\la) \geq 0$ and $\eta'(\la) \geq 0$, provided $\delta_{\min}^{(2)} \leq \varrho \leq \delta_{\max}^{(2)}$.

Let $f(\varrho) = \mathfrak{q}(\eta) - \mathfrak{q}(\eta')$, and recall from the proof of Lemma~\ref{lem:NonNegEtas} that $\eta(\la) = P(\la) + Q(\la) \varrho$ and $\eta'(\la) = P'(\la) + Q'(\la) \varrho$.  Write $\eta_\la$ for $\eta(\la)$, etc., to save space.
Expand $f(\varrho)$ and use Lemma~\ref{lem:sameLength}: 
\begin{align*}
f(\varrho) &= \sum_{\la \in \OO} (\eta_\la^2 - \eta_\la^{'2}) \\
&= \sum_{\la \in \OO} (P_\la^2 - P^{'2}_\la) + 2 \sum_{\la \in \OO} (P_\la Q_\la - P'_\la Q'_\la) \varrho + \sum_{\la \in \OO} (Q_\la^2 - Q_\la^{'2}) \varrho^2 .
\end{align*}

Examine $\sum_{\la \in \OO} (Q_\la^2 - Q_\la^{'2})$ closely.  If $((q-1)/t) \breve{b}^2 \neq t b_0^2 + q c_m(b,b)$, then from \eqref{eqn:defineOmegas}, $Q_\la = g_{\la,\ell_0} + g_{\la, \al^m \ell_0}$ and $Q'_\la = g_{\la, \ell_0} + g_{\la, \ell_\infty}$, where $g_{\la, \la'}$ is the $\la, \la'$-entry of $W^{-1}$.  
Lemma~\ref{lem:blockTypes}, applied to $b$,  says that every row or column of $B=W^{-1}$ has the same entries, just in different positions. This means that the values of $\sum_{\la \in \OO} g_{\la,\la'}^2$ are independent of $\la'$.  Thus  $\sum_{\la \in \OO} (Q_\la^2 - Q_\la^{'2}) = 2\sum_{\la \in \OO} (g_{\la,\ell_0} g_{\la, \al^m \ell_0} - g_{\la, \ell_0} g_{\la, \ell_\infty})$.

Because $W$ and $W^{-1}$ are symmetric matrices, we recognize the sum $\sum_{\la \in \OO} (Q_\la^2 - Q_\la^{'2})$ as twice the difference of the $\ell_0, \al^m \ell_0$-entry of $W^{-2}$ minus the $\ell_0, \ell_{\infty}$-entry of $W^{-2}$.  By Lemma~\ref{lem:ABentries}, those entries are $t b_0^2 + q c_m(b,b)$ and $((q-1)/t) \breve{b}^2$, respectively. Thus $\sum_{\la \in \OO} (Q_\la^2 - Q_\la^{'2}) \neq 0$.  

Using the same argument, if $t b_0^2 + q c_{m_1}(b,b) \neq t b_0^2 + q c_{m_2}(b,b)$, we recognize the sum $\sum_{\la \in \OO} (Q_\la^2 - Q_\la^{'2})$ as twice the difference of the $\ell_0, \al^{m_1} \ell_0$-entry of $W^{-2}$ minus the $\ell_0, \al^{m_2} \ell_0$-entry of $W^{-2}$.  Again, by Lemma~\ref{lem:ABentries}, those entries are $t b_0^2 + q c_{m_1}(b,b)$ and $t b_0^2 + q c_{m_2}(b,b)$, respectively. Thus $\sum_{\la \in \OO} (Q_\la^2 - Q_\la^{'2}) \neq 0$.  

In either case, we conclude that $f(\varrho)$ is a nontrivial quadratic function, with at most two rational zeros.  By having $\varrho$, with $\delta_{\min}^{(2)} \leq \varrho \leq \delta_{\max}^{(2)}$, avoid the zeros of $f$, we see that $\mathfrak{q}(\eta) \neq \mathfrak{q}(\eta')$.
\end{proof}

\begin{rem}
It can be shown, under the hypotheses of Theorem~\ref{thm:differentQs}, that $\varrho=1$ is a double zero of $f$.  
\end{rem}

Let us examine the hypotheses of Theorem~\ref{thm:differentQs} in the situation where $t=2$.  As $t \mid (q-1)$, $q$ must be odd when $t=2$.  
The weight $w$ has only two nonzero values, $w_1=w(\al)$ and $w_2 = w(\al^2)$.  In fact, $w_1 \neq w_2$; otherwise, $\sym(w) = \F_q^\times$, and we would have $t=1$ instead of $t=2$.  We conclude that $w$ is nondegenerate.  Indeed, from \eqref{eqn:CirculantA},
\[  \mathcal{W} = \begin{bmatrix}
w_2 & w_1 \\
w_1 & w_2
\end{bmatrix} . \]
Then $\det \mathcal{W} = w_2^2 - w_1^2 = (w_1 + w_2)  (w_1 - w_2) \neq 0$.

Because $1 \leq m \leq \lfloor t/2 \rfloor$, the only possible value of $m$ is $m=1$ when $t=2$.  Now consider $((q-1)/t) \breve{w}^2 = q c_1(w,w)$, with $t=2$.

Expand $((q-1)/2) \breve{w}^2 = q c_1(w,w)$ to get
$(q-1) (w_1^2 + 2 w_1 w_2 + w_2^2) = 4 q w_1 w_2$, which is symmetric in $w_1$ and $w_2$.  Set $x=w_2/w_1$ and manipulate to get
\[  w_1 \left( (q-1) - 2(q+1)x + (q-1) x^2 \right) = 0 . \] 
The symmetry of the equation in $w_1, w_2$ implies that the solutions in $x$ will be reciprocals.  Solving for $x$ yields
\[ x = \frac{q+1 \pm 2 \sqrt{q}}{q-1} . \]
When $q$ is not a square, these values of $x$ (and hence at least one of $w_1, w_2$) are irrational.  But the values are rational when $q$ is a square.

Here are some small values, remembering that $q$ must be odd (and a prime power).
\begin{equation}  \label{eqn:t2Equality}
\begin{array}{c|c|c|c|c}
q & 9 & 25 & 49 & 81 \\ \hline
x, 1/x & 2, 1/2 & 3/2, 2/3 & 4/3, 3/4 & 5/4, 4/5
\end{array}
\end{equation}

To summarize:
\begin{prop}  \label{prop:nonSquares}
Suppose that $q$ is an odd power of an odd prime.  Let $w$ be any integer-valued weight on $\F_q$, with $w(0)=0$, $w(r)>0$ for $r \neq 0$.  If $t=(q-1)/\size{\sym(w)} = 2$, then $w$ is nondegenerate and satisfies $((q-1)/2) \breve{w}^2 \neq q c_1(w,w)$.
\end{prop}

\section{Final arguments}
In this section, we finish the proof of the main theorem, Theorem~\ref{thm:MainV2}.

\begin{proof}[Proof of Theorem~\textup{\ref{thm:MainV2}}]
Apply Theorem~\ref{thm:differentQs}: pick a rational value of $\varrho$ that satisfies $\delta_{\min}^{(2)} \leq \varrho \leq \delta_{\max}^{(2)}$ and is such that $\mathfrak{q}(\eta) \neq \mathfrak{q}(\eta')$, where $\eta, \eta'$ are the solutions to $W \eta = \omega$ and $W \eta' = \omega'$ coming from \eqref{eqn:defineOmegas}.  Because $\omega$ and $\omega'$ are rational vectors and $W$ is a matrix  with integer entries that is invertible over $\Q$, the vectors $\eta$ and $\eta'$ have (nonnegative) rational entries.

Choose a sufficiently large integer $N$ so that $N\eta$ and $N\eta'$ have (nonnegative) integer entries.  Define linear codes $C$ and $D$ by using $N\eta$ and $N\eta'$, respectively, as their multiplicity functions.  Their lists of orbit weights are, respectively, $N \omega = W(N\eta)$ and $N \omega' = W(N \eta')$.  Because $W$ is an integer matrix with positive entries and $N \eta$, $N \eta'$ are nonnegative integer vectors, the entries of $N\omega$, $N \omega'$ are nonnegative integers.  By the format of \eqref{eqn:defineOmegas}, we have
\[  \wwe_C = \wwe_D = 1 + 2 \size{\sym(w)} t^{N \varrho} + (t(q+1)-2) \size{\sym(w)} t^N .  \]
That is, there are two $\sym(w)$-orbits having weight $N \varrho$, and the remaining orbits have weight $N$.

Now consider the dual codes.  We claim that $A_{2 \mathring{w}}(C^\perp) \neq A_{2 \mathring{w}}(D^\perp)$.  
By Lemma~\ref{lem:2MathringInequality}, 
\begin{align*}
A_{2 \mathring{w}}(C^\perp) &- A_{2 \mathring{w}}(D^\perp) \\
= &\asing_{2 \mathring{w}}(C^\perp) - \asing_{2 \mathring{w}}(D^\perp) 
+ \adble_{2 \mathring{w}}(C^\perp) - \adble_{2 \mathring{w}}(D^\perp) .
\end{align*}

Because $\wwe_C = \wwe_D$, Proposition~\ref{prop:singletonsDoNotDistinguish} implies that $\asing_s(C^\perp) = \asing_s(D^\perp)$, for all $s$.
Remember that $\varrho$ was chosen so that $\mathfrak{q}(\eta) \neq \mathfrak{q}(\eta')$.  Then, using Lemma~\ref{lem:sameLength}, we see that $\mathfrak{q}(N\eta) - \mathfrak{q}(N\eta') = N^2 (\mathfrak{q}(\eta) - \mathfrak{q}(\eta')) \neq 0$.  By Proposition~\ref{prop:2w-weights}, we have $\adble_{2 \mathring{w}}(C^\perp) - \adble_{2 \mathring{w}}(D^\perp) \neq 0$.  Thus $A_{2 \mathring{w}}(C^\perp) \neq A_{2 \mathring{w}}(D^\perp)$. 
\end{proof}

\section{Symmetrized enumerators}
In this short section, we present a brisk summary of symmetrized enumerators given by a group action.  This will allow us to calculate the $w$-weight enumerator in many examples.  Proofs and additional details can be found in \cite{MR3336966},  \cite{wood:duality}, \cite{wood:turkey}, or \cite{MR4963840}.

We will focus on linear codes over a finite field $\F_q$.  The field $\F_q$ has size $q = p^l$ for some prime $p$ and positive integer $l$.  The \emph{absolute trace} $\tr \colon  \F_q \ra \F_p$ is given by 
\[  \tr(r) = r + r^p + r^{p^2} + \ldots + r^{p^{l-1}}, \quad r \in \F_q . \]
The absolute trace is a nonzero homomorphism of $\F_p$-vector spaces.  Fix a primitive $p$th root of unity $\zeta \in \C$.  Define $\chi \colon  \F_q \ra \C^\times$ by $\chi(r) = \zeta^{\tr(r)}$;  $\chi$ is a generating character for $\F_q$.

Fix a subgroup $H \subseteq \F_q^\times$.  Fix a primitive element $\al$ in $\F_q$, and let $t$ be the smallest positive integer such that $\al^t$ generates $H$.  The subgroup $H$ acts on $\F_q$ by multiplication.  The orbits of $H$ are: $P_0=\{0\}$ and the cosets $P_i=\al^i H$, $i=1, 2, \ldots, t$, inside $\F_q^\times$.  Note that $P_t = H$.  The orbits form a partition $\mathcal{P}_H = (P_i)_{i=0}^t$ of $\F_q$.  For $r \in \F_q$, write $\nu(r) = i$ when $r \in P_i$.

The \emph{Kravchuk matrix} $K$ associated to $\mathcal{P}_H$ is the $(t+1) \times (t+1)$ matrix defined for $i,j = 0, 1, \ldots, t$ by
\begin{align}
K_{0,j} &= \size{P_j} , \notag \\
K_{i,j} &= \sum_{s \in P_j} \chi(\al^i s) , \quad i=1, 2, \ldots, t . \label{eqn:Kravchuk}
\end{align}

The \emph{symmetrized enumerator} $\se_C$ of a linear code $C \subseteq \F_q^n$ is the partition enumerator of the partition $\mathcal{P}_H$.  It is an element of the polynomial ring $\C[Z_0, Z_1, \ldots, Z_t]$ defined by 
\begin{equation}  \label{eqn:defnOfSymEnum}
\se_C(Z_0, Z_1, \ldots, Z_t) = \sum_{c \in C} \prod_{k=1}^n Z_{\nu(c_k)} ,
\end{equation}
where $c = (c_1, c_2, \ldots, c_n) \in C$.  The product in \eqref{eqn:defnOfSymEnum} keeps track of how many entries of $c$ belong to each $H$-orbit.

The symmetrized enumerator satisfies the MacWilliams identities:
\begin{equation}  \label{eqn:MWidsForSymEnum}
\se_{C^\perp}(Z_0, \ldots, Z_t) = \frac{1}{\size{C}} \se_C(\mathcal{Z}_0, \mathcal{Z}_1, \ldots, \mathcal{Z}_t)|_{\mathcal{Z}_i \leftarrow \sum_{j=0}^t K_{i,j} Z_j} .
\end{equation}

Suppose $w$ is an integer-valued weight on $\F_q$, with $w(0)=0$ and $w(r)>0$ for nonzero $r$.  Let $H = \sym(w)$, and form the partition $\mathcal{P}_H$, as above.   Then the $w$-weight enumerator $\wwe_C$ of a linear code $C \subseteq \F_q^n$ can be recovered from the symmetrized enumerator $\se_C$ via specializing variables $Z_0 \leadsto 1$ and $Z_i \leadsto y^{w(\al^i)}$ for $i=1, 2, \ldots, t$.  
I.e.,
\begin{equation}  \label{eqn:WweFromSymEnum}
\wwe_C(y) = \se_C(Z_0, Z_1, \ldots, Z_t)|_{Z_0 \leftarrow 1, Z_i \leftarrow y^{w(\al^i)}} .
\end{equation}
Note that any element of the orbit $P_i = \al^i H$ has the form $r = \al^i h$ for some $h \in H$.  Then $w(r) = w(\al^i)$ because $h \in H = \sym(w)$.

\section{Examples}

In this section, we provide a few examples that illustrate Theorem~\ref{thm:MainV2} and its hypotheses.  All complicated calculations were carried out using SageMath \cite{sagemath}.  The first example carries out explicit calculations that mirror the proof of Theorem~\ref{thm:MainV2}.  

\begin{ex}  \label{ex:F5}
Consider the Euclidean weight $\we$ on $\F_5$:
\[ \begin{array}{c|ccccc}
r & 0 & 1 & 2 & 3 & 4 \\ \hline
\we(r) & 0 & 1 & 4 & 4 & 1
\end{array} . \]
Note that $\sym(\we) = \{ \pm 1\} = \{1,4\}$.  We choose $\al=2$, so that $t=2$.  Then the matrix $\mathcal{W}$ of \eqref{eqn:CirculantA} is 
\[  \mathcal{W} = \begin{bmatrix}
1 & 4 \\ 4 & 1
\end{bmatrix} . \]
Then $\det \mathcal{W} = -15 \neq 0$, so the weight $\we$ is nondegenerate.  (In general, the Euclidean weight is nondegenerate over any $\Z/m\Z$ by \cite{MR3947347}.)
As in Example~\ref{ex:F5Amatrix}, we see that 
\[  W = \begin{bmatrix}
1 & 4 & 0 & 0 & 4 & 1 & 1 & 4 & 4 & 1 & 1 & 4 \\
4 & 1 & 0 & 0 & 1 & 4 & 4 & 1 & 1 & 4 & 4 & 1 \\ 
0 & 0 & 1 & 4 & 1 & 4 & 1 & 4 & 1 & 4 & 1 & 4 \\
0 & 0 & 4 & 1 & 4 & 1 & 4 & 1 & 4 & 1 & 4 & 1 \\
4 & 1 & 1 & 4 & 0 & 0 & 1 & 4 & 4 & 1 & 4 & 1 \\
1 & 4 & 4 & 1 & 0 & 0 & 4 & 1 & 1 & 4 & 1 & 4 \\
1 & 4 & 1 & 4 & 1 & 4 & 4 & 1 & 4 & 1 & 0 & 0 \\
4 & 1 & 4 & 1 & 4 & 1 & 1 & 4 & 1 & 4 & 0 & 0 \\
4 & 1 & 1 & 4 & 4 & 1 & 4 & 1 & 0 & 0 & 1 & 4 \\
1 & 4 & 4 & 1 & 1 & 4 & 1 & 4 & 0 & 0 & 4 & 1 \\
1 & 4 & 1 & 4 & 4 & 1 & 0 & 0 & 1 & 4 & 4 & 1 \\
4 & 1 & 4 & 1 & 1 & 4 & 0 & 0 & 4 & 1 & 1 & 4
\end{bmatrix} .
\]
Then $ -75 W^{-1}$ equals
\[ \left[ \begin{array}{rrrrrrrrrrrr}
1 & -4 & 6 & 6 & -4 & 1 & 1 & -4 & -4 & 1 & 1 & -4 \\
-4 & 1 & 6 & 6 & 1 & -4 & -4 & 1 & 1 & -4 & -4 & 1 \\
6 & 6 & 1 & -4 & 1 & -4 & 1 & -4 & 1 & -4 & 1 & -4 \\
6 & 6 & -4 & 1 & -4 & 1 & -4 & 1 & -4 & 1 & -4 & 1 \\
-4 & 1 & 1 & -4 & 6 & 6 & 1 & -4 & -4 & 1 & -4 & 1 \\
1 & -4 & -4 & 1 & 6 & 6 & -4 & 1 & 1 & -4 & 1 & -4 \\
1 & -4 & 1 & -4 & 1 & -4 & -4 & 1 & -4 & 1 & 6 & 6 \\
-4 & 1 & -4 & 1 & -4 & 1 & 1 & -4 & 1 & -4 & 6 & 6 \\
-4 & 1 & 1 & -4 & -4 & 1 & -4 & 1 & 6 & 6 & 1 & -4 \\
1 & -4 & -4 & 1 & 1 & -4 & 1 & -4 & 6 & 6 & -4 & 1 \\
1 & -4 & 1 & -4 & -4 & 1 & 6 & 6 & 1 & -4 & -4 & 1 \\
-4 & 1 & -4 & 1 & 1 & -4 & 6 & 6 & -4 & 1 & 1 & -4
\end{array} \right]. \]
Note that the row and column sums of $W$ equal $25$, while those of $W^{-1}$ equal $1/25$, as expected from Corollary~\ref{cor:rowSumsOfW}.  
The matrix $W^2$ has $2 \times 2$ blocks $\left[\begin{smallmatrix}
85 & 40 \\ 40 & 85
\end{smallmatrix}\right]$ along the diagonal, with all other entries equal to $50$.
Similarly, the matrix $5625 W^{-2}$ has $2 \times 2$ blocks $\left[\begin{smallmatrix}
157 & 32 \\ 32 & 157
\end{smallmatrix}\right]$ along the diagonal, with all other entries equal to $-18$.
Because $c_0(w,w) = 17$, $c_1(w,w) = 8$, and $\breve{w}=5$, we see that $W^2$ matches the form of Lemma~\ref{lem:ABentries}.  Similarly for the $\sym(w)$-invariant function $b$ that satisfies $B = -75 W^{-1}$: $b(0)=6$, $b(1)=1$, $b(2)=-4$.  Then $\breve{b}=-3$, $c_0(b,b)=17$, and $c_1(b,b)=-8$.

As in the proof of Theorem~\ref{thm:differentQs}, set
\begin{align*}
\omega &= \langle 1,1,\varrho, \varrho, 1,1,1,1,1,1,1,1 \rangle , \\
\omega' &= \langle \varrho,1,\varrho, 1, 1,1,1,1,1,1,1,1 \rangle .
\end{align*}
Solve $W\eta = \omega$ and $W \eta' = \omega'$; then clear denominators using $N = 75$:
\begin{align*}
75 \eta = &\langle -12\varrho + 15, -12\varrho + 15, 3\varrho, 3\varrho, 3\varrho, 3\varrho, 3\varrho, 3\varrho, 3\varrho, 3\varrho, 3\varrho, 3\varrho \rangle, \\
75 \eta' = &\langle -7\varrho + 10, -2\varrho + 5, -7\varrho + 10, -2\varrho + 5, 3\varrho, 3\varrho, \\
&\quad -2\varrho + 5, 8\varrho - 5, 3\varrho, 3\varrho, -2\varrho + 5, 8\varrho - 5 \rangle .
\end{align*}
All the entries will be nonnegative if $5/8 \leq \varrho \leq 5/4$.  
One calculates that $f(\varrho)= \mathfrak{q}(\eta) - \mathfrak{q}(\eta') = 100(\varrho -1)^2$.  The leading coefficient is related to entries of $W^{-2}$: $100 = 2(32 -(-18))$.

Use both extreme points of $\varrho$, in turn, suitably normalized:
\begin{align*}
\eta_{C_-} & = \langle 4, 4, 1, 1, 1, 1, 1, 1, 1, 1, 1, 1 \rangle , \\
\omega_{C_-} &= \langle 40, 40, 25, 25, 40, 40, 40, 40, 40, 40, 40, 40 \rangle , \\
\eta_{D_-} & = \langle 3, 2, 3, 2, 1, 1, 2, 0, 1, 1, 2, 0 \rangle , \\
\omega_{D_-} &= \langle 25, 40, 25, 40, 40, 40, 40, 40, 40, 40, 40, 40 \rangle ; \\
\eta_{C_+} & = \langle 0, 0, 3, 3, 3, 3, 3, 3, 3, 3, 3, 3 \rangle , \\
\omega_{C_+} &= \langle 60, 60, 75, 75, 60, 60, 60, 60, 60, 60, 60, 60 \rangle , \\
\eta_{D_+} & = \langle 1, 2, 1, 2, 3, 3, 2, 4, 3, 3, 2, 4 \rangle , \\
\omega_{D_+} &= \langle 75, 60, 75, 60, 60, 60, 60, 60, 60, 60, 60, 60 \rangle .
\end{align*}
The linear codes $C_-$ and $D_-$ have length $18$, while $C_+$ and $D_+$ have length $30$; cf., Lemma~\ref{lem:sameLength}.  
All the nonzero orbits have size $2$, so
\begin{align*}
\ewe_{C_-} = \ewe_{D_-} &= 1 + 4 y^{25} + 20 y^{40} , \\
\ewe_{C_+} = \ewe_{D_+} &= 1 + 20 y^{60} + 4 y^{75} .
\end{align*}
Note that the ratios of the weights are $25/40=5/8$ and $75/60 = 5/4$, by design.
From Lemma~\ref{lem:sameLength}, Lemma~\ref{lem:2MathringInequality}, and Proposition~\ref{prop:2w-weights}, we have 
\begin{align*}
A_2(C_-^\perp) - A_2(D_-^\perp) &= 24-20 = 4, \\
A_2(C_+^\perp) - A_2(D_+^\perp) &= 60-56 = 4.
\end{align*} 

The symmetrized enumerators \eqref{eqn:defnOfSymEnum} for the linear codes are:
\begin{align*}
\se_{C_-} &= Z_0^{18} + 4 Z_0^8 Z_1^5 Z_2^5 + 20 Z_0^2 Z_1^8 Z_2^8 , \\
\se_{D_-} &= Z_0^{18} + 4 Z_0^5 Z_1^9 Z_2^4 + 16 Z_0^2 Z_1^8 Z_2^8 + 4 Z_0^5 Z_1^4 Z_2^9 ; \\
\se_{C_+} &= Z_0^{30} + 20 Z_0^6 Z_1^{12} Z_2^{12} + 4 Z_1^{15} Z_2^{15} , \\
\se_{D_+} &= Z_0^{30} + 4 Z_0^3 Z_1^{16} Z_2^{11} + 16 Z_0^6 Z_1^{12} Z_2^{12} + 4 Z_0^3 Z_1^{11} Z_2^{16} .
\end{align*}
By specializing $Z_0 \leadsto 1$, $Z_1 \leadsto y^4$, and $Z_2 \leadsto y$, the symmetrized enumerators yield the Euclidean weight enumerators via \eqref{eqn:WweFromSymEnum}.

The Kravchuk matrix \eqref{eqn:Kravchuk} for the MacWilliams identities for $\se$ is
\[  K = \begin{bmatrix}
1 & 2 & 2 \\
1 & (\sqrt{5}-1)/2 & -(\sqrt{5}+1)/2 \\
1 & -(\sqrt{5}+1)/2 & (\sqrt{5}-1)/2
\end{bmatrix} . \]
Then \eqref{eqn:MWidsForSymEnum}, followed by specialization \eqref{eqn:WweFromSymEnum}, yield
\begin{align*}
\wwe_{C_-^\perp} &= 1 + 24 y^2 + 296 y^3 +1900 y^4 + 10760 y^5 + \cdots , \\
\wwe_{D_-^\perp} &= 1 + 20 y^2 + 296 y^3 + 1956 y^4 + 10760 y^5 + \cdots ; \\
\wwe_{C_+^\perp} &= 1 + 60 y^2 + 1260 y^3 + 17880 y^4 + 182772 y^5 + \cdots , \\
\wwe_{D_+^\perp} &= 1 + 56 y^2 + 1304 y^3 + 17720 y^4 + 182816 y^5 + \cdots .
\end{align*}
\end{ex}

The next two examples discuss some of the weights and finite fields appearing in \eqref{eqn:t2Equality}. 
In one case, over $\F_9$, we find that the MacWilliams identities hold.  In the other case, over $\F_{25}$, we find that, while Theorem~\ref{thm:MainV2} does not apply, the linear codes given by Lemma~\ref{lem:NonNegEtas} nonetheless demonstrate that duality in not respected.

\begin{ex}  \label{ex:F9t2}
Let $q=9$.  The polynomial $x^2 - x -1$ is irreducible over $\F_3$.  Set $\F_9 = \F_3[\al]/(\al^2 - \al -1)$.  The multiplicative group $\F_9^\times$ is a cyclic group of order $8$, and it is generated by $\al$.  That is, $\al$ is a primitive element:
\[  \begin{array}{c|cccccccc}
i & 0 & 1 & 2 & 3 & 4 & 5 & 6 & 7 \\ \hline
\al^i & 1 & \al & 1 + \al & 1 + 2 \al & 2 & 2 \al & 2 + 2 \al & 2 + \al 
\end{array} . \]

There is a unique subgroup $H \subset \F_9^\times$ with $\size{H}=4$:  $H = \{ 1, 1+\al, 2, 2+2\al \}$; $H$ is generated by $\al^2$, so that $t=2$.  A weight $w$ on $\F_9$ with $\sym(w) = H$ will have two nonzero positive values, say $w_2$ on $H$ and $w_1$ on the coset $\al H = \{ \al, 1+2\al, 2\al, 2+\al \}$.  Note that $w_1 \neq w_2$.  Indeed, if $w_1=w_2$, then $\sym(w) = \F_9^\times$ is strictly larger than $H$.

The weight $w$ is nondegenerate.  The matrix $\mathcal{W}$ of \eqref{eqn:CirculantA} is
\[  \mathcal{W} = \begin{bmatrix}
w_2 & w_1 \\
w_1 & w_2
\end{bmatrix} , \]
whose determinant is nonzero: $\det \mathcal{W} = w_2^2 - w_1^2 \neq 0$.

Of special interest from \eqref{eqn:t2Equality} is the case where $w_1=1, w_2 = 2$ (or the reverse: $w_1=2, w_2=1$).  As expected, $((q-1)/t) \breve{w}^2 = 36= q c_1(w,w)$, so Theorem~\ref{thm:MainV2} does not apply.

Define two $\F_3$-linear isomorphisms $f,g \colon  \F_9 \ra \F_3^2$, using $\{1, \al \}$ as an $\F_3$-basis of $\F_9$:
\begin{align*}
f(x+y\al) = (x+y,x) , \quad\quad g(x+y\al) = (y, x-y) .
\end{align*}
Here are their values, with $(a,b) \in \F_3^2$ written as $ab$:
\[  \begin{array}{c|ccccccccc}
x+y\al & 0 & 1 & \al & 1 + \al & 1 + 2 \al & 2 & 2 \al & 2 + 2 \al & 2 + \al \\
f(x+y\al) & 00 & 11 & 10 & 21 & 01 & 22 & 20 & 12 & 02 \\
g(x+y\al) & 00 & 01 & 12 & 10 & 22 & 02 & 21 & 20 & 11 
\end{array} . \]
When $w$ has $w_1=1, w_2=2$, then $f$ is an isometry between $\F_9$ with the weight $w$ and $\F_3^2$ with the Hamming weight from $\F_3$.  When the values are reversed, i.e., $w_1=2, w_2=1$, it is $g$ that is an isometry to the Hamming weight.  

Suppose two vectors $v, v' \in \F_9^n$ satisfy $v \cdot v'=0$ using the standard dot product over $\F_9$.  Write $v = (x_1+y_1 \al, \ldots, x_n + y_n \al)$ and similarly for $v'$.  Note that the absolute trace of $r \in \F_9$ is $\tr(r) = r + r^3$, so that $\tr(1)=2$, $\tr(\al) = 1$, and $\tr(\al^2) = 0$.  Then
\begin{align*}
0 = \tr(v \cdot v') &= \sum_{i=1}^n \left( x_i x'_i + (x_i y'_i +y_i x'_i ) \al + y_i y'_i \al^2 \right) \\
&= \sum_{i=1}^n \left(  x_i x'_i \tr(1) + (x_i y'_i + y_i x'_i ) \tr(\al) + y_i y'_i \tr(\al^2) \right) \\
&= \sum_{i=1}^n \left(2 x_i x'_i + x_i y'_i + y_i x'_i \right) \\
&= \begin{cases}
\sum_{i=1}^n \left( (x_i+y_i)x'_i + x_i (x'_i + y'_i) \right) \\
\sum_{i=1}^n \left( y_i y'_i - (x_i - y_i)(x'_i - y'_i) \right) .
\end{cases}
\end{align*}

While the last expressions are not the standard dot product on $\F_3^{2n}$, it can be shown that the MacWilliams identities for the Hamming weight still hold for those inner products.  Thus, the two weights on $\F_9$ considered here inherit the MacWilliams identities and therefore respect duality.
\end{ex}

\begin{ex}
Let $q=25$ and $t=2$.  The polynomial $x^2-x-3$ is irreducible over $\F_5$.  Use this polynomial to present $\F_{25} = \F_5[\be]/(\be^2-\be-3)$.  The multiplicative group $\F_{25}^\times$ is a cyclic group of order $24$, and $\be$ is a generator (primitive element).  There is a unique subgroup $H$ of order $12$; $H$ is generated by $\be^2 = \be + 3$.  A weight $w$ with $\sym(w) = H$ will have two nonzero positive values, say $w_2$ on $H$ and $w_1$ on the other coset $\be H$.  As argued in Example~\ref{ex:F9t2}, $w_1 \neq w_2$.

Of special interest from \eqref{eqn:t2Equality} is the case where $w_1=2, w_2 = 3$ (or the reverse: $w_1=3, w_2=2$).  SageMath-assisted calculations give $\delta_{\min}^{(2)} = 5/6$ and $\delta_{\max}^{(2)} = 25/24$.  Further calculations, with $\varrho \in \{ 5/6, 25/24 \}$ in \eqref{eqn:defineOmegas}, yield two pairs of linear codes: $C,D$ of length $62$, and $C', D'$ of length $250$, with the following $w$-weight enumerators:
\begin{align*}
\wwe_C &= 1 + 24 t^{125} + 600 t^{150} , \\
\wwe_D &= 1 + 24 t^{125} + 600 t^{150} , \\
\wwe_{C^\perp} &= 1 + 360 t^4 + 1464 t^5 + 114120 t^6 + \cdots , \\
\wwe_{D^\perp} &= 1 + 360 t^4 + \hphantom{1}864 t^5 + 115920 t^6 + \cdots ; \\
\wwe_{C'} &= 1 + 600 t^{600} + 24 t^{625} , \\
\wwe_{D'} &= 1 + 600 t^{600} + 24 t^{625} , \\
\wwe_{C^{'\perp}} &= 1 + 6000 t^4 + 15000 t^5 + 7116000 t^6 + \cdots , \\
\wwe_{D^{'\perp}} &= 1 + 6000 t^4 + 14400 t^5 + 7117800 t^6 + \cdots .
\end{align*}
For both values of $\varrho$, we see that the dual codes differ at $A_5$.  Thus, the weight $w$ does not respect duality.

The values of $A_5$ for a dual code can be computed directly from the multiplicity function $\eta$ of the primal code $C$.  A vector of weight $5$ must be a doubleton with one nonzero entry in $H$ and the other nonzero entry in $\be H$.   A doubleton will be a dual codeword if $x \be^{i_1} \ell_{\mu_1} + y \be^{i_2} \ell_{\mu_2} = 0$, where $x, y$ are the nonzero entries of the doubleton, appearing in positions $\be^{i_1} \ell_{\mu_1}, \be^{i_2} \ell_{\mu_2}$, respectively. 
Because the lines $\ell_\mu$ are linearly independent, we must have $\mu_1 = \mu_2$.  Then $x \be^{i_1} + y \be^{i_2} = 0$, or $x = - y \be^{i_2-i_1}$.  

Because $t=2$, $i_1, i_2 \in \{ 0, 1 \}$.  In order for $x,y$ to be in different cosets, $i_2-i_1$ must be odd.  Thus for every column equal to $\ell_\mu$ and every column equal to $\be \ell_\mu$, take any $y \in \F_{25}^\times$ and set $x = - y \be$.  Then $x \ell_\mu + y \be \ell_\mu = 0$.  This gives the count
\[  A_5(C^\perp) = \sum_{\mu} 24 \eta(\ell_\mu) \eta(\be \ell_\mu) . \]
\end{ex}

Our final example discusses the power weights $\wl$.
\begin{ex}
Consider the power weights $\wl$ of Example~\ref{ex:PowerWeights} over $\Z/m\Z$.  Because $w^{(0)}$ is the Hamming weight, the MacWilliams identities hold for $w^{(0)}$ over any $\Z/m\Z$, $m \geq 2$, \cite{wood:duality}.
The power weight $w^{(1)}$ equals the Lee weight $\wlee$.  We know that the MacWilliams identities for $\wlee$ hold over $\Z/m\Z$ for $m=2, 3, 4$, and $\wlee$ does not respect duality over $\Z/m\Z$ for $m \geq 5$, \cite{MR4119402}.

Assume $\ell \geq 1$.  Then the symmetry group 
$H = \sym(\wl) = \{ \pm 1\}$.  Consider first the situation where $m = p$, a prime.  If $p=2$ or $p=3$, all the power weights $\wl$ equal the Hamming weight.  Once again, we know that the MacWilliams identities hold.  So now assume $p \geq 5$.

For $p \geq 5$, each power weight $\wl$, $\ell \geq 1$, satisfies three of the five hypotheses of Theorem~\ref{thm:MainV2}, namely, $-1 \in \sym(\wl)$, $\sym(\wl) \neq \F_p^\times$ (this is where $p \geq 5$ is needed), and the minimum value $\mathring{w}^{(\ell)} = 1$ is achieved on $H = \{ \pm 1\}$.
Remaining are the hypotheses on nondegeneracy and on $q c_1(w,w)$, etc., having different values.

The Lee weight $\wlee=w^{(1)}$ and the Euclidean weight $\we = w^{(2)}$ are known to be nondegenerate over $\F_p$ \cite{DLW-deux}; in fact, they are nondegenerate over all $\Z/m\Z$ \cite{MR3947347}.  For $\ell \geq 3$, there is only numerical evidence.  For example, SageMath computations reveal that the matrices $\mathcal{W}^{(\ell)}$ of \eqref{eqn:CirculantA} have maximal rank for $\ell = 3,4, \ldots, 10$, for the first $100$ primes.

Now turn to the hypothesis on $q c_1(w,w)$, etc., having different values.  Because $\size{H}=2$, $t=(p-1)/2$.  If $p=5$, then $t=2$, and only $5 c_1(\wl,\wl)$ and $2 (\breve{w}^{(\ell)})^2$ are available.  The expressions are:
\begin{align*}
5 c_1(\wl,\wl) &= 5 \cdot 2^{\ell+1} , \\  
2 (\breve{w}^{(\ell)})^2 &= 2(1+2^\ell)^2 . 
\end{align*}
As $\ell \geq 1$, these values are different: the first is divisible by $4$, while the second is congruent to $2 \bmod 4$.

Likewise for $p=7$, where $t=3$: again, only  $7 c_1(\wl,\wl)$ and $2 (\breve{w}^{(\ell)})^2$ are available.  The expressions now are:
\begin{align*}
7 c_1(\wl,\wl) &= 7 (2^\ell + 3^\ell + 6^\ell) , \\
2 (\breve{w}^{(\ell)})^2 &= 2(1+2^\ell + 3^\ell)^2 .
\end{align*}
Again, the values are different: the first is odd; the second is even.

For $p  \geq 11$, we have $t \geq 5$, so $p c_2(\wl,\wl)$ is also available (and possibly other $p c_j(\wl,\wl)$).  Again, there is numerical evidence.  SageMath computations indicate that $c_1(\wl,\wl) \neq c_2(\wl,\wl)$ for $\ell = 1, 2, \ldots, 10$, for the first $1000$ primes greater than $7$.

Once one knows that some $\wl$ satisfies the hypotheses of Theorem~\ref{thm:MainV2} for all primes $p \geq 5$, one can then extend the results to all $\Z/m\Z$, $m \geq 4$ ($m \geq 5$ for $\ell=1$ only) by following the corresponding portion of the proof in \cite{MR4119402}.  This will require finding examples of linear codes that do not respect duality for $m=4, 6, 9$ ($m=6, 8, 9$ for $\ell=1$).
\end{ex}

We conclude with a conjecture.
\begin{conj}
The power weights $\wl$, $\ell \geq 2$, do not respect duality over any $\Z/m\Z$, $m \geq 4$.
\end{conj}

\def\cprime{$'$} \def\cprime{$'$} \def\cprime{$'$} \def\cprime{$'$}
  \def\cprime{$'$}
\providecommand{\bysame}{\leavevmode\hbox to3em{\hrulefill}\thinspace}
\providecommand{\MR}{\relax\ifhmode\unskip\space\fi MR }
\providecommand{\MRhref}[2]{%
  \href{http://www.ams.org/mathscinet-getitem?mr=#1}{#2}
}
\providecommand{\href}[2]{#2}

\end{document}